\documentclass[oneside,english]{amsart}
\usepackage[T1]{fontenc}
\usepackage[latin9]{inputenc}
\usepackage{geometry}
\usepackage{amstext}
\usepackage{amsthm}
\usepackage{amssymb}
\usepackage{setspace}
\makeatletter
\numberwithin{equation}{section}
\numberwithin{figure}{section}
\theoremstyle{plain}
\newtheorem{thm}{\protect\theoremname}[section]
\theoremstyle{plain}
\newtheorem{cor}[thm]{\protect\corollaryname}
\theoremstyle{plain}
\newtheorem{lem}[thm]{\protect\lemmaname}
\theoremstyle{definition}
\newtheorem{defn}[thm]{\protect\definitionname}
\theoremstyle{remark}
\newtheorem{claim}[thm]{\protect\claimname}

\makeatother

\usepackage{babel}
\providecommand{\claimname}{Claim}
\providecommand{\corollaryname}{Corollary}
\providecommand{\definitionname}{Definition}
\providecommand{\lemmaname}{Lemma}
\providecommand{\theoremname}{Theorem}

\begin{document}
\title{Hypercontractivity on high dimensional expanders:\\
\footnotesize Approximate Efron-Stein Decompositions for {\Large$\varepsilon$}-product spaces}
\author{Tom Gur, Noam Lifshitz, Siqi Liu}
\thanks{Tom Gur, University of Warwick, tom.gur@warwick.ac.uk, supported by
the UKRI Future Leaders Fellowship MR/S031545/1.}
\thanks{Noam Lifshitz, Hebrew University of Jerusalem, noam.lifshitz@mail.huji.ac.il,
supported in part by ERC advanced grant 834735.}
\thanks{Siqi Liu, UC Berkeley, sliu18@berkeley.edu.}
\begin{abstract}
We prove hypercontractive inequalities on high dimensional expanders.
As in the settings of the $p$-biased hypercube, the symmetric group,
and the Grassmann scheme, our inequalities are effective for global
functions, which are functions that are not significantly affected
by a restriction of a small set of coordinates. As applications, we
obtain Fourier concentration, small-set expansion, and Kruskal--Katona
theorems for high dimensional expanders. Our techniques rely on a
new approximate Efron--Stein decomposition for high dimensional link
expanders.
\end{abstract}

\maketitle

\section{Introduction}

High-dimensional expanders (HDX) are sparse simplicial complexes with
strong structural properties. More accurately, a simplicial complex
$X$ is a $\epsilon$-HDX (or an $\epsilon$-link expander) if the
$1$-skeleton of each link of the complex $X$ is a spectral expander
graph whose second-largest eigenvalue is bounded by $\epsilon$. In
recent years, HDX have received much attention in theoretical computer
science \cite{KaufmanM17,parzanchevski2017mixing,KaufmanO18a,lubotzky2018high,liu2019high,KaufmanO20,golowich2021improved},
finding applications in property testing \cite{KaufmanL14,DinurK17,gotlib2019testing},
coding theory \cite{DinurHKNT19,dikstein2020locally}, statistical
physics \cite{anari2021spectral,chen2021optimal,anari2021entropic},
complexity theory \cite{alev2019approximating,anari2019log,dinur2020explicit,hopkins2020high},
and beyond. Notably, very recently the study of HDX led to a breakthrough
in quantum computing, breaking the $\sqrt{n}$ distance in quantum
LDPC \cite{evra2020decodable}, as well as to a resolution of one
of the most important questions in coding theory, namely, the first
construction of $O(1)$-query asymptotically good locally testable
codes \cite{c3LTC}.

In this work, we focus on analysis of Boolean functions on high dimensional
expanders, whose systematic study was recently initiated by Dikstein
et al. \cite{DiksteinDFH18}. This continues a long line of investigation
of Fourier analysis of Boolean functions on extended domains beyond
the Boolean hypercube, such as the Boolean slice \cite{ODonnellK13,Filmus16,FilmusM16,FilmusKMW18},
the Grassmann scheme \cite{DinurKKMS18a,KhotMS18,EKL21}, the symmetric
group \cite{filmus2020hypercontractivity,Filmus2021+,DFLLV2021},
the $p$-biased cube \cite{ellis2019biased,lifshitz2019noise,Filmus2021exposition},
and the multi-slice \cite{FOW2019,braverman2021invariance}. The foregoing
extended domains arise naturally throughout theoretical computer science,
and indeed, the study of analysis of Boolean functions on extended
domains has recently led to a breakthrough regarding the unique games
conjecture \cite{KhotMS17,DinurKKMS18,DinurKKMS18a,KhotMS18}.

Hypercontractive inequalities are amongst the most powerful technical
tools in Fourier analysis, yielding a plethora of applications in
algorithms, complexity, learning theory, statistical physics, social
choice, and beyond (see \cite{ODonnell} and references therein).
Loosely speaking, such statements assert that functions of low Fourier
degree are ``well behaved'' in terms of their distribution around
their mean. Concretely, in the Boolean hypercube, the simplest example
of a hypercontractive inequality is Bonami's lemma, which states that
for every function $f\colon\{0,1\}^{n}\to\mathbb{{R}}$ of Fourier
degree at most $d$, it holds that $\|f\|_{4}\leq\sqrt{{3}^{d}}\|f\|_{2}$.

Alas, in the setting of high dimensional expanders, where the domain
is not a product space and the induced measure is biased, general
strong hypercontractivity cannot hold. The heart of the problem is
that some highly local functions, such as dictators (i.e., $f(x)=x_{i}$),
provide strong counterexamples to hypercontractivity. A similar phenomenon
also occurs in several prominent extended domains, such as the $p$-biased
cube and the Grassmann scheme.

Fortunately, as observed in the setting of the $p$-biased cube \cite{keevash2019hypercontractivity},
all of the aforementioned examples are local, in the sense that a
small number of coordinates can significantly influence the output
of the function. This led to the definition of `global' functions.
For Boolean valued functions, these are functions wherein a small
number of coordinates can change the output of the function only with
a negligible probability. For real valued functions, this is captured
by the $2$-norm remaining roughly the same when restricting O(1)
coordinates of the input. More precisely, consider the setting of
a general product measure. Let $\left(V_{i},\mu_{i}\right)$ be probability
spaces, let $V_{S}=\prod_{i\in S}V_{i}$ and equip $V_{S}$ with the
product measure, which we denote by $\mu_{S}.$ Every function $f\in L^{2}\left(V_{\left[k\right]},\mu\right)$
is equipped with an orthogonal decomposition $\sum_{S\subseteq\left[n\right]}f^{=S}$
known as the Efron--Stein decomposition. The function $f^{=S}$ in
the Efron--Stein decomposition plays a similar role to the function
$\hat{f}\left(S\right)\chi_{S}$ in the Boolean cube. Using that analogy
we write 
\[
f^{\le d}=\sum_{\left|S\right|\le d}f^{=S},
\]
 and $f$ is said to be of degree $d$ if $f=f^{\le d}$. Keevash
et al. \cite{KLLM} introduced the following notions. The \emph{Laplacians
of $f$} are given by 
\[
L_{S}\left[f\right]=\sum_{T\supseteq S}\left(-1\right)^{\left|T\right|}f^{=T}.
\]
For $x\in V_{S}$ the \emph{derivatives} are given by restricting
the laplacians 
\[
D_{S,x}f=L_{S}\left[f\right]\left(x,\cdot\right),
\]
and the $\left(S,x\right)$\emph{-influence }of $f$ is defined as
\[
I_{S,x}\left[f\right]=\|D_{S,x}\left[f\right]\|_{2}^{2}.
\]
In this setting, a function $f$ is $\left(r,\delta\right)$\emph{-global}
if $\|f\left(x,\cdot\right)\|_{2,\mu_{\left[n\right]\setminus S}}^{2}\leq\delta$
for each $\left|S\right|\le r.$ We remark that here, being $\left(r,\delta\right)$-global
for a small $\delta>0$ is, in a sense, equivalent to having $I_{S,x}\left[f\right]\le\delta'$
for a small $\delta'$ for all $\left|S\right|\le r$ and all $x$.
In fact, $\delta,\delta'$ can be taken to be within a factor of $2^{r}$
of one another.

In \cite{KLLM}, it was shown that if $f\in L^{2}\left(V,\mu\right)$
is of degree $d$, then the following hypercontractive inequality
holds:

\begin{equation}
\|f\|_{4}^{4}\le1000^{d}\sum_{S}\mathbb{E}_{x\sim\mu_{S}}I_{S,x}\left[f\right]^{2}.\label{eq:hypercontractivity}
\end{equation}
This allowed them to deduce if a function $f$ of degree $d$ is $\left(d,\delta\right)$\emph{-global},
then 
\[
\|f\|_{4}^{4}\le\delta8000^{d}\|f\|_{2}^{2}.
\]
 Here when setting $\delta=100\|f^{\le d}\|_{2}^{2}$ one gets the
statement $\|f\|_{4}\le C^{d}\|f\|_{2},$ which replicates the behavior
in the Boolean cube. Moreover, the statement is useful even for larger
values of $\delta$.

In this work, we raise the following question.

\medskip{}

\begin{quote}
\begin{center}
\emph{Does hypercontractivity hold for high dimensional expanders?}
\par\end{center}

\end{quote}
\medskip{}

\subsection{Main results}

We answer the question above in the affirmative. Namely, our main
contribution is a hypercontractive inequality for functions on the
$k$-faces of an $\epsilon$-HDX. We denote by $X(k)$ the $k$-faces
of a simplicial complex $X$, and denote by $\mu$ the uniform measure
on its $k$-faces. We define the influences $I_{S,x}^{\le d}$ and
the degree restriction operator $(\cdot)^{\le d}$ analogously to
their definition on the $p$-biased cube (see Section \ref{sec:Efron=002013Stein}
for precise definition). We then prove the following hypercontractive
statement for high dimensional expanders in the spirit of (\ref{eq:hypercontractivity}).
\begin{thm}
Let $X$ be an $\epsilon$-HDX, and let $f\in L^{2}\left(X(k),\mu\right)$.
We have 
\[
\|f^{\le d}\|_{4}^{4}\le20^{d}\sum_{\left|S\right|\le d}\left(4d\right)^{\left|S\right|}\mathbb{E}_{x\sim\mu_{S}}I_{S,x}^{\le d}\left[f\right]^{2}+O_{k}\left(\epsilon^{2}\right)\|f\|_{2}^{2}\|f\|_{\infty}^{2}.
\]
\end{thm}

In the setting of $\epsilon$-HDX, we say that a function $f$ is
$\left(d,\delta\right)$-global if for each $\left|S\right|\le d$,
we have $\|f\left(x,\cdot\right)\|_{L^{2}\left(V_{x},\mu_{x}\right)}\le\delta.$
We show that we can bound the infinity norm of global functions and
obtain the following strong hypercontractive inequality for global
functions on $\epsilon$-HDX.
\begin{cor}
For each $\zeta,d,k>0,$ there exists $\epsilon_{0}=\epsilon_{0}\left(\zeta,k,d\right),\delta_{0}=\delta_{0}\left(\zeta,d\right)$,
such that the following holds. Let $\epsilon\le\epsilon_{0},\delta\le\delta_{0}$,
let $X$ be an $\epsilon$-HDX, and let $f\in L^{2}\left(X(k),\mu\right)$.
If $f$ is $\left(d,\delta\right)$-global, then we have 
\[
\|f^{\le d}\|_{4}^{4}\le\zeta\|f\|_{2}^{2}.
\]
\end{cor}

We remark that, in fact, we prove our results in a slightly more general
setting, to which we refer as $\epsilon$-product measures. See Section
\ref{sec:Proving-hypercontractivity} for details.

\subsection{Applications}

As corollaries of our hypercontractive inequality for high dimensional
expanders, we obtain several applications, which we discuss below.
See Section \ref{sec:Applications} for more details.

\subsubsection{Fourier spectrum concentration theorem}

Fourier concentration results are widely useful in complexity theory
and learning theory. Our first application is a Fourier concentration
theorem for HDX. Namely, the following theorem shows that global Boolean
functions on $\epsilon$-HDX are concentrated on the high degrees,
in the sense that the $2$-norm of the restriction of a function to
its low-degree coefficients only constitutes a tiny fraction of its
total $2$-norm.
\begin{thm}
\label{thm:concentration on the high degrees} For each $\zeta,d,k>0,$
there exists $\epsilon_{0}=\epsilon_{0}\left(\zeta,k,d\right),\delta_{0}=\delta_{0}\left(\zeta,d\right)$,
such that the following holds. Let $\epsilon\le\epsilon_{0},\delta\le\delta_{0}$,
let $X$ be an $\epsilon$-HDX, and let $f\colon X(k)\to\left\{ 0,1\right\} $
be $\left(d,\delta\right)$-global. Then 
\[
\|f^{\le d}\|_{2}^{2}\le\zeta\|f\|_{2}^{2}.
\]
\end{thm}

\subsubsection{Small set expansion theorem}

Small set expansion is a fundamental property that is prevalent in
combinatorics and complexity theory. In the setting of the $\rho$-noisy
Boolean hypercube, the small set expansion theorem of Kahn, Kalai,
and Linial \cite{KahnKL88} gives an upper bound on $\mathrm{{Stab}_{\rho}}(1_{A})=\langle1_{A},T_{\rho}1_{A}\rangle$
for indicators $1_{A}$ of small sets $A$. The noise stability $\mathrm{{Stab}_{\rho}}(1_{A})$
captures the probability that a random edge $\left(x,y\right)$ of
the $\rho$-noisy hypercube has both its endpoints in $A$. Hence,
an inequality of the form $\mathrm{{Stab}_{\rho}}(1_{A})\le\zeta\|1_{A}\|_{2}^{2}$
for an arbitrarily small $\zeta$ and sufficiently small $A$ implies
that that small sets are expanding in the sense that the random walk
makes you leave them with probability $\ge1-\zeta$. Our second application
is a small set expansion theorem for global functions on $\epsilon$-HDX,
captured via bounding the natural noise operator in this setting.
Let $\rho\in\left(0,1\right)$ be a noise-rate parameter. The noise
operator is given by 
\[
\mathrm{T}_{\rho}f\left(x\right):=\sum_{S\subseteq\left[k\right]}\rho^{\left|S\right|}\left(1-\rho\right)^{k-\left|S\right|}\mathbb{E}_{y\sim\mu}\left[f\left(y\right)|y_{S}=x_{S}\right].
\]
 In other words, $\mathrm{T}_{\rho}$ corresponds to the random walk
that starts with $x$ chooses a $\rho$-biased random $S\subseteq\left[k\right]$,
keeps $x_{S}$, and re-randomises $x$ given $x_{S}$. Our small set
expansion theorem tells us that if we start with a small subset $A\subseteq X\left(k\right)$
and we apply one step of the random walk, then we leave $A$ with
probability $0.99$.
\begin{thm}
\label{thm:small set expansion}For each $\zeta,d,k>0,$ there exists
$\epsilon_{0}=\epsilon_{0}\left(\zeta,k,d\right),\delta_{0}=\delta_{0}\left(\zeta,d\right)$,
such that the following holds. Let $\epsilon\le\epsilon_{0},\delta\le\delta_{0}$,
and let $X$ be an $\epsilon$-HDX. If $f\colon V_{\left[k\right]}\to\left\{ 0,1\right\} $
is $\left(d,\delta\right)$-global, then 
\[
\|\mathrm{T}_{\rho}f\|_{2}^{2}\le\zeta\|f\|_{2}^{2}.
\]
\end{thm}

\subsubsection{Kruskal--Katona theorem}

Our last application is an analogue of the Kruskal--Katona theorem
in the setting of high dimensional expanders. The Kruskal-Katona theorem
is a fundamental and widely-applied result in extremal combinatorics,
which gives a lower bound on the size of the lower shadow $\partial\left(A\right)$
of a $k$-uniform hypergraph $A$ on $n$ vertices. The \emph{lower
shadow }is defined to be the family of all $(k-1)$-sets that are
contained in an edge of $A$. More generally, if $A\subseteq X\left(k\right)$,
then we similarly let $\partial\left(A\right)$ be the family of all
$k-1$-faces that are contained in a $k$-face of $A$.

Filmus et al. \cite{filmus2020hypercontractivity} used their hypercontractivity
theorem to prove a stability result for the Kruskal--Katona theorem.
We prove a similar stability result for $\epsilon$-HDX.
\begin{thm}
Let $X$ be an $\epsilon$-HDX, for a sufficiently small $\epsilon>0$.
Let $\delta\le\left(200d\right)^{-d},$ and let $A\subseteq X\left(k-1\right)$
be $\left(d,\delta\right)$-global. Then 
\[
\mu\left(\partial\left(A\right)\right)\ge\mu\left(A\right)\left(1+\frac{d}{2k}\right).
\]
\end{thm}

\subsection{Techniques}

Conceptually, one can view the theory of expanders and pseudorandom
graphs in the following perspective: Given a pseudorandom regular
graph $G=\left(V,E\right)$ and $\left(x,y\right)\sim E$, the goal
is to show that $x,y$ behave similarly to independent random variables
$x,y\sim V$, i.e., as an approximation of a product space.

In the theory of high dimensional expanders, we are given a distribution
$\mu$ on $(k+1)$-tuples by choosing a random $k$-face $\left(x_{1},\ldots,x_{k+1}\right)$
of a sparse simplicial complex, and the goal is again to show that
the variables $\{x_{i}\}$ approximately behave as though they were
independent. Thus, our main objective is to generalise results from
the product space setting, where the $x_{i}$'s are independent, to
the setting of HDX, where we only have local spectral information
about the links. However, such a generalisation yields significant
challenges.

One of the fundamental tools for studying the product space setting
is the aforementioned Efron--Stein decomposition. Its role in the
analysis of product spaces is that it allows us to easily generalise
techniques from the Boolean cube by replacing the Fourier expression
$\hat{f}\left(S\right)\chi_{S}$ with the function $f^{=S}.$

Our high-level proof strategy is to develop new Efron--Stein decompositions
for HDX. We show that despite the more involved setting, and despite
the fact that we only have mere local spectral information, we can
still obtain similar structural properties as in product spaces. We
now list a few of the challenges that we are facing, which require
fundamentally new ideas and techniques.

Dikstein et al. \cite{DiksteinDFH18} gave a decomposition of the
form $f=\sum_{d=0}^{k}f^{=d}$. We provide a new decomposition $\left\{ f^{=S}\right\} _{S\subseteq\left[k\right]}$
such that $f=\sum_{S\subseteq\left[k\right]}f^{=S},$ and despite
not having orthogonality, we can still show that the inner product
$\left\langle f^{=S},f^{=T}\right\rangle $ is negligible compared
to $\|f\|_{2}^{2}.$ This allows us to generalise the Laplacians,
derivatives and influences, but we have to deal with the following
problems:
\begin{itemize}
\item Let $\mathcal{F}\subseteq\left[k\right]$ be a small set. We would
like to say that $g=\sum_{S\in\mathcal{F}}f^{=S}$ is supported on
$\mathcal{F},$ but we have no way of knowing that looking at $\left\{ g^{=S}\right\} _{S\subseteq\left[k\right]},$
as $g^{=S}$ may be nonzero even for $S\notin\mathcal{F}.$ This leads
to the problem of how to even define the degree of a function. We
would like to say that $f^{\le d}:=\sum_{\left|S\right|\le d}f^{=S}$
is of degree at most $d$, and that $f$ is of degree $d$ if $f=f^{\le d}$.
Alas, according to this definition the function $f^{\le d}$ is not
of degree $d$.
\item We can and do define the derivatives $D_{S,x}$ to be the restrictions
of the Laplacians. In the product case the derivatives decrease the
degree by $\left|S\right|$, and this is a very desirable property
as our proof goes by induction on $d$. However, this is no longer
true in the HDX setting.
\item We may define the influences by taking $2$-norms of the derivatives.
However, now it is no longer true that having small influences is
equivalent to being global. This leads us to the following problem
which is the source for all of the difficulty.
\item The spectral information tells us that HDX should behave similarly
to product spaces with respect to the $L^{2}$-norm. However, we care
about $L_{4}$ information when bounding $\|f\|_{4}$, and we deal
with $L_{\infty}$-hypothesis as the globalness notion is about \textbf{all}
the restrictions. There is no reason for HDX to behave well with respect
to $L_{4}$ and even more so for $L_{\infty}.$
\end{itemize}
At first, the above, and especially the last point, seem as fundamental
barriers to this approach.

Nevertheless, we overcome this barrier by developing an alternative
notion, which we call the \emph{approximate Efron--Stein decomposition}.
Our new notion has the following properties that fix all of the above
problems.
\begin{itemize}
\item If $\left\{ f_{S}\right\} _{S\subseteq\left[k\right]}$ is an approximate
Efron--Stein decomposition, then crucially, $\left\{ f_{S}\right\} _{S\in\mathcal{F}}$
is an approximate Efron--Stein decomposition for $\sum_{S\in\mathcal{F}}f^{=S}$.
\item If $f$ is approximately of degree $d$, in the sense that $\left\{ f_{S}\right\} $
is an approximate Efron-Stein decomposition for $f$, then the derivative
$D_{S,x}\left[f\right]$ may be $L_{4}$-approximated by $D_{S,x}\left[f\right]^{\le d-\left|S\right|}$
.
\item We find a way of proving an inequality of the form 
\[
\mathbb{E}_{x\sim\mu_{S}}I_{S,x}^{2}\left[f\right]\le\delta\mathbb{E}_{x\sim\mu_{S}}\left[I_{S,x}\right],
\]
 without having the traditional hypothesis $\max_{x}I_{S,x}\left[f\right]\le\delta$
at our disposal.
\item We show that we may move freely between different approximate Efron--Stein
decomposition up to a small $L_{4}$-norm error term.
\end{itemize}
We believe that our approximate Efron--Stein decomposition provides
the desired comfortable platform for analysing functions on HDX in
the same way one would analyze a product space.

See Section \ref{sec:Efron=002013Stein} for a detailed exposition
of our approximate Efron--Stein decomposition, and see Section \ref{sec:overrview}
for a more detailed proof overview of our main hypercontractivity
results, which build on the aforementioned decomposition.

\subsection{Related work}

Simultaneously and independently to this work, Bafna, Hopkins, Kaufman,
and Lovett \cite{BHKS} also obtained hypercontractive inequalities
for high dimensional expanders. We remark that while the main hypercontractive
inequalities in both papers achieve essentially the same parameters,
the techniques are completely different. Namely, in \cite{BHKS} the
proof strategy follows the approach of analogous results in the setting
of the Grassmann graph, whereas our approach generalises Efron--Stein
decompositions and hypercontractivity for general product spaces.
We further note that our approximate Efron--Stein decomposition extends
approximate Fourier decompositions that appeared in several recent
works \cite{KaufmanO18a,KaufmanO20,DiksteinDFH18,alev2019approximating,jeronimo2021near}.

\subsection{Organisation}

The rest of the paper is organised as follows. We start in Section
\ref{sec:Recalling-globalness}, where we recall the notions of hypercontractivity
and globalness in general product spaces, as well as provide an alternative
proof of a slightly weaker hypercontractive inequality that is more
amenable for generalisation to non-product spaces. In Section \ref{sec:eps-spaces},
we present the framework of $\epsilon$-product spaces, of which high
dimensional expanders are a special case, and we also define key operators
in this setting and show some basic properties they satisfy. Next,
in Section \ref{sec:Efron=002013Stein}, which is introducing a new
approximate Efron--Stein decomposition and developing a framework
for proving hypercontractivity results using this decomposition. Then,
in Section \ref{sec:overrview}, we give a detailed proof overview
of our hypercontractive inequalities for high dimensional expanders,
which build on the foregoing framework. In Section \ref{sec:notions},
we define the notions of laplacians, derivatives and influences in
the setting of $\epsilon$-measures, give bounded approximated Efron--Stein
decompositions related to the Laplacians, define globalness, and show
that it implies small influences.. Then, we provide the full proof
of our main hypercontractivity results in Section \ref{sec:Proving-hypercontractivity}.
Finally, in Section \ref{sec:Applications}, we show how to derive
the applications from our hypercontractive inequalities.

\subsection*{Acknowledgments}

We are grateful to Alessandro Chiesa for participating in early stages
of this research. We thank Yuval Filmus and Dor Minzer for insightful
discussions. We further thank Yuval Filmus for insightful comments
on a previous version of this work.

\section{Recalling globalness and hypercontractivity in the product space
setting\label{sec:Recalling-globalness}}

We begin by recalling the Efron--Stein decomposition, as well as
derivatives and Laplacians in the setting of general product spaces,
and state the hypercontractivity inequalities for product spaces that
were shown in \cite{KLLM}. We then give a proof, inspired by \cite{EKL21},
of a slightly weaker hypercontractivity inequality that we will later
generalise to approximate product spaces.

\subsection{Efron-Stein decomposition}

Let $\left(V_{1},\mu_{1}\right),\ldots,\left(V_{k},\mu_{k}\right)$
be a probability space. Let $\mu$ be the corresponding product measure
$\mu_{1}\otimes\cdots\otimes\mu_{k}$. For a set $S\subseteq\left[k\right]$,
we write $V_{S}=\prod_{i\in S}V_{i},$ and we write $\text{\ensuremath{\mu_{S}} for the product measure }\,\mu_{S}=\bigotimes_{i\in S}\mu_{i}.$
The \emph{Efron--Stein decomposition} is a decomposition of $L^{2}\left(V_{\left[k\right]},\mu\right)$
into $2^{k}$ orthogonal spaces $\left\{ W_{S}\right\} _{S\subseteq\left[k\right]}.$
Every function $f\in L^{2}\left(V_{\left[k\right]},\mu\right)$ can
then be decomposed as $f=\sum_{S\subseteq\left[k\right]}f^{=S},$
where $f^{=S}$ is the projection of $f$ to $W_{S}.$ The Efron--Stein
decomposition is characterised by the orthogonality of $\left\{ W_{S}\right\} $,
the fact that $\sum_{S}W_{S}=L^{2}\left(V,\mu\right)$, and the fact
that the space $W_{S}$ is composed of functions depending only on
$S.$

The functions $f^{=S}$ also have an explicit formula for $x\in V_{S}$,
where we denote
\[
A_{S}f\left(x\right)=\mathbb{E}_{y\sim\left(V_{\overline{S}},\mu_{\overline{S}}\right)}\left[f\left(x,y\right)\right],
\]
 where $\bar{{S}}=[k]\setminus S$. We then write 
\[
f^{=S}=\sum_{T\subseteq S}\left(-1\right)^{\left|S\setminus T\right|}A_{T}f.
\]
 The function $A_{S}\left[f\right]$ then has the following neat Efron--Stein
decomposition 
\[
A_{S}\left[f\right]=\sum_{T\subseteq S}f^{=T}.
\]
 See \cite[Chapter 8]{ODonnell} for more details.

\subsection{Notations}

We write $a=b\pm\epsilon$ to indicate that $a\in\left(b-\epsilon,b+\epsilon\right).$
We use $a\le O\left(b\right)$ to denote that the inequality holds
up to an absolute constant, and $a\le O_{k}\left(b\right)$ to denote
that the inequality holds up to a constant only depending on $k$.

\subsection{Derivatives and Laplacians}

Let $\mu=\mu_{1}\otimes\cdots\otimes\mu_{k}$ be a product measure.
Let $f\in L^{2}\left(V_{\left[k\right]},\mu\right),$ $S\subseteq\left[n\right]$.
The \emph{Laplacian} is given by the formula 
\[
L_{S}\left[f\right]=\sum_{T\supseteq S}f^{=T}=\sum_{T\subseteq S}\left(-1\right)^{\left|T\right|}A_{\left[k\right]\setminus T}f.
\]
For $S\subseteq\left[n\right]$ and $x\in V_{S}$ the \emph{derivative}
$D_{S,x}\in L^{2}\left(V_{\overline{S}},\mu_{\overline{S}}\right)$
is defined by 
\[
D_{S,x}f=L_{S}\left[f\right]\left(x,\cdot\right).
\]
For convenience, we also write $D_{\varnothing}f=f.$ The $\left(S,x\right)$\emph{-influence
}of $f$ is defined as 
\[
I_{S,x}\left[f\right]=\|D_{S,x}\left[f\right]\|_{2}^{2}.
\]
 This includes the case $S=\varnothing$, where we have $I_{\varnothing}\left[f\right]:=\|f\|_{2}^{2}.$

We now state a few facts from \cite{KLLM} that we generalise. The
following lemma, which appears in \cite{KLLM}, shows that the notion
of small influences corresponds to small $2$-norms of the restriction
of $f$.
\begin{lem}
Suppose that $I_{S,x}\left[f\right]\le\delta$ for each set $S$ of
size at most $r.$ Then $\|f\left(x,\cdot\right)\|_{2,\mu_{\left[n\right]\setminus S}}^{2}\le\delta4^{r}$
for each $S$ and $x\in V_{S}$. Conversely, if $\|f\left(x,\cdot\right)\|_{2,\mu_{\left[n\right]\setminus S}}^{2}\le\delta$
for each $\left|S\right|\le r$ and $x\in V_{S}$, then $I_{S,x}\left[f\right]\le\delta4^{r}$
for each $S$ of size at most $r$ and $x\in V_{S}.$
\end{lem}

For the above reason they gave the following definition.
\begin{defn}
A function $f$ is said to be $\left(r,\delta\right)$\emph{-global}
if $I_{S,x}\left[f\right]\le\delta$ for each $\left|S\right|\le r.$
\end{defn}

The \emph{degree }of a function is the largest $S$, such that $f^{=S}\ne0$.
The derivatives decrease the degrees for the following reason.
\begin{lem}
$D_{S,x}\left[f^{=T}\right]$ is $0$ unless $S\subseteq T$, and
if $S\subseteq T$, then 
\[
D_{S,x}\left[f^{=T}\right]\in W^{T\setminus S}.
\]
\end{lem}

Consequently, if $f=\sum_{\left|S\right|\le d}f^{=S}$ is of degree
$d$, then $D_{S,x}\left[f\right]$ is of degree $d-\left|S\right|.$

\subsection{Hypercontractivity}

The following result is by \cite{KLLM}.
\begin{thm}
\label{thm:KLLM}If $f\in L^{2}\left(V,\mu\right)$ is of degree $d$,
then

\[
\|f\|_{4}^{4}\le1000^{d}\sum_{S}\mathbb{E}_{x}I_{S,x}\left[f\right]^{2}.
\]
\end{thm}

To show the implication of the theorem for global functions they use
the following inequality.
\begin{lem}
\label{lem: sum of low degree influences is bounded } 
\[
\sum_{S}\mathbb{E}_{x}I_{S,x}\left[f\right]\le2^{d}\|f\|_{2}^{2}.
\]
\end{lem}

\begin{proof}
The right hand side is equal to 
\[
\sum_{S}\|L_{S}\left[f\right]\|_{2}^{2}=\sum_{S}\sum_{T\supseteq S,\left|T\right|\le d}\|f^{=T}\|_{2}^{2}\le2^{d}\sum_{T}\|f^{=T}\|_{2}^{2}=2^{d}\|f\|_{2}^{2}.
\]
\end{proof}
Combining Theorem \ref{thm:KLLM} and Lemma \ref{lem: sum of low degree influences is bounded },
we obtain the following corollary.
\begin{cor}
\label{cor:low degree functions have small 4-norms} If $f$ of degree
$d$ is $\left(d,\delta\right)$-global. Then $\|f^{\le d}\|_{4}^{4}\le\delta2000^{d}\|f\|_{2}^{2}.$
\end{cor}

\begin{proof}
We have 
\begin{align*}
\|f^{\le d}\|_{4}^{4} & \le1000^{d}\sum_{S\subseteq\left[n\right]}\mathbb{E}_{x\sim\mu_{S}}\|I_{S,x}\left[f\right]\|_{2}^{4}\\
 & \le\delta1000^{d}\sum_{S}\mathbb{E}_{x\sim\mu_{S}}\|I_{S,x}\left[f\right]\|_{2}^{2}\\
 & \le\delta2000^{d}\|f\|_{2}^{2}.
\end{align*}
\end{proof}

\subsection{An alternative proof of hypercontractivity on product spaces}

We give an alternative proof of the following slightly weaker version
of Theorem \ref{thm:KLLM}. The proof is inspired by a future work
by Ellis, Kindler, and the second author \cite{EKL21}, who show that
the same idea works in the Grassmann setting. In this paper we show
that it generalises to HDX as well.
\begin{thm}
\label{thm:Our version of hypercontrativity-product case} Let $f\in L^{2}\left(V,\mu\right)$
be of degree $d.$ Then 
\[
\|f\|_{4}^{4}\le2\cdot9^{d}\sum_{\left|T\right|\le d}\left(9d\right)^{\left|T\right|}\mathbb{E}_{x}\left[I_{S,x}\left[f\right]^{2}\right].
\]
\end{thm}

\subsubsection{Proof overview.}

Before providing the full proof, we first describe the high-level
approach for proving Theorem \ref{thm:Our version of hypercontrativity-product case}.
The strategy is to first show a lemma that gives the following bound
\begin{equation}
\|f\|_{4}^{4}\le C^{d}\|f\|_{2}^{4}+\sum_{S\subseteq\left[n\right]}\left(4d\right)^{\left|S\right|}\|L_{S}\left[f\right]\|_{4}^{4},\label{eq:inductive lemma}
\end{equation}

for a constant $C$. Using this lemma, we can give an inductive proof
by first noting that $\|L_{S}\left[f\right]\|_{4}^{4}=\mathbb{E}_{x}\|D_{S,x}\|_{4}^{4}$,
and then applying induction using the fact that $D_{S,x}$ is of degree
$d-\left|S\right|$. Finally, using the fact that $D_{S,x}D_{T,y}=D_{S\cup T,\left(x,y\right)}$,
we can get our desired hypercontractive statement.

Hence, the key step is to prove the aforementioned lemma. To this
end, we first use the fact that

\[
\mathbb{E}\left[f^{4}\right]=\sum_{S}\|\left(f^{2}\right)^{=S}\|_{2}^{2}.
\]

We then expand the summands of $\left(f^{2}\right)^{=S}$ as sums
of terms of the form $\left(f^{=T_{1}}f^{=T_{2}}\right)^{=S}.$ Next,
we note that the nonzero terms either satisfy $T_{1}\cap T_{2}\cap S\ne\varnothing$
or satisfy $T_{1}\Delta T_{2}=S$. Terms of the first kind are cancelled
out by $L_{i}\left[f\right]^{4}$ for an $i\in T_{1}\cap T_{2}\cap S$
on the right hand side of (\ref{eq:inductive lemma}). (The terms
$\|L_{S}\left[f\right]\|_{4}^{4}$ appear because of over counting,
which we resolve by inclusion exclusion.) Terms of the latter kind
correspond to the situation in the Boolean cube where $f^{=T}=\hat{f}\left(T\right)\chi_{T}$
and $\chi_{T}\chi_{S}=\chi_{T\Delta S}.$ We then upper bound $\|\left(f^{=T_{1}}f^{=T_{2}}\right)^{=S}\|_{2}$
by $\|f^{=T_{1}}\|_{2}\|f^{=T_{2}}\|_{2}$. This allows us to translate
the problem of upper bounding the terms of the first kind to the problem
of upper bounding the $4$-norm of a low degree function on the Boolean
cube. Namely, the function 
\[
\sum_{\left|T\right|\le d}\|f^{=T}\|_{2}\chi_{T}.
\]
Finally, we use hypercontractivity to upper bound the $4$-norm by
its $2$-norm, which is equal to the $2$-norm of $f$. This concludes
the proof overview.

\subsubsection{Proof of hypercontractivity on product spaces}

We now give a formal proof of Theorem \ref{thm:Our version of hypercontrativity-product case}.
We shall first need the following key lemma, which admits the inductive
approach.
\begin{lem}
\label{lem:Inductive approach} Let $f\in L^{2}\left(V,\mu\right)$
be of degree $d.$ Then 
\[
\frac{1}{2}\|f\|_{4}^{4}\le9{}^{d}\|f\|_{2}^{4}+\sum_{T\ne\varnothing}\left(4d\right)^{\left|T\right|}\|L_{T}\left[f\right]\|_{4}^{4}.
\]
\end{lem}

We are now ready to prove the lemma.
\begin{proof}[Proof of Lemma \ref{lem:Inductive approach}]
 By Parseval we have 
\[
\|f\|_{4}^{4}=\sum\|\left(f^{2}\right)^{=S}\|_{2}^{2}.
\]
 We bound each term $\|\left(f^{2}\right)^{=S}\|_{2}^{2}$ individually.
By expanding and using the linearity of the $\cdot^{=S}$ operator
we have
\[
\left(f^{2}\right)^{=S}=\sum_{T_{1},T_{2}}\left(f^{=T_{1}}f^{=T_{2}}\right)^{=S}.
\]
 We now divide the pairs $\left(T_{1},T_{2}\right)$ into three sums.
\begin{enumerate}
\item We let $I_{1}$ be the set of pairs $\left(T_{1},T_{2}\right)$ such
that $T_{1}\cap T_{2}\cap S\ne\varnothing.$ If $i$ is in $T_{1}\cap T_{2}\cap S$,
then the summand $\left(f^{=T_{1}}f^{=T_{2}}\right)^{=S}$ appears
as a summand when expanding $\left(L_{i}\left[f\right]^{2}\right)^{=S}.$
This explains the role of the Laplacians in the right hand side.
\item We let $I_{2}$ be the set of pairs such that $T_{1}\Delta T_{2}=S.$
These kind of pairs have a similar behavior to the one in the Boolean
cube. There $f^{=S}=\hat{f}\left(S\right)\chi_{S}$ and 
\[
f^{=S}f^{=T}=\hat{f}\left(S\right)\hat{f}\left(T\right)\chi_{S\Delta T}.
\]
 We show that the contribution from the pairs in $I_{2}$ is $\le C^{d}\|f\|_{2}^{2}.$
\item We let $I_{3}=\left(T_{1},T_{2}\right)$ such that either $\left(T_{1}\Delta T_{2}\right)\setminus S\ne\varnothing$
or $S\setminus\left(T_{1}\cup T_{2}\right)\ne\varnothing$. We show
that in this case $\left(f^{=T_{1}}f^{=T_{2}}\right)^{=S}=0.$
\end{enumerate}
It is easy to verify that each pair $\left(T_{1},T_{2}\right)$ belongs
to at least one of the sets $I_{1},I_{2},I_{3}$. We additionally
have $I_{1}\cap I_{2}=\varnothing.$

\subsection*{Upper bounding the contribution from $I_{1}$}

Let us start by upper bounding the contribution from pairs corresponding
to $I_{1}$. For a nonempty $T\subseteq S$ write $I_{1}\left(T\right)$
for the pairs $\left(T_{1},T_{2}\right)$, such that $T_{1}\cap T_{2}\supseteq T.$
Then 
\[
\left(L_{T}\left[f\right]^{2}\right)^{=S}=\sum_{\left(T_{1},T_{2}\right)\in I_{1}\left(T\right)}\left(f^{=T_{1}}f^{=T_{2}}\right)^{=S}.
\]
 Now $I_{1}=\bigcup_{i\in S}I_{1}\left(i\right),$ so as a multiset
inclusion-exclusion shows that we have
\[
I_{1}=\sum_{T\subseteq S}\left(-1\right)^{\left|T\right|-1}\bigcap_{i\in T}I_{1}\left(i\right)=\sum_{T\subseteq S}\left(-1\right)^{\left|T\right|-1}I_{1}\left(T\right).
\]

We therefore have the equality: 
\[
\sum_{\left(T_{1},T_{2}\right)\in I_{1}}\left(f^{=T_{1}}f^{=T_{2}}\right)^{=S}=\sum_{T\subseteq S,T\ne\varnothing}\left(-1\right)^{\left|T\right|-1}\left(L_{T}\left[f\right]^{2}\right)^{=S}.
\]
 By the triangle inequality and Cauchy--Schwarz, we obtain that 
\begin{align*}
\|\sum_{\left(T_{1},T_{2}\right)\in I_{1}}\left(f^{=T_{1}}f^{=T_{2}}\right)^{=S}\|_{2}^{2} & \le\left(\sum_{i=1}^{\left|S\right|}\binom{\left|S\right|}{i}\left(4\left|S\right|\right)^{-i}\right)\left(\sum_{T\subseteq S}\left(4\left|S\right|\right)^{\left|T\right|}\|\left(L_{T}\left[f\right]^{2}\right)^{=S}\|_{2}^{2}\right)\\
 & \le\sum_{T\subseteq S}\left(4\left|S\right|\right)^{\left|T\right|}\|\left(L_{T}\left[f\right]^{2}\right)^{=S}\|_{2}^{2}.
\end{align*}
 Summing over all $S$ we have 
\[
\sum_{S}\|\sum_{\left(T_{1},T_{2}\right)\in I_{1}}\left(f^{=T_{1}}f^{=T_{2}}\right)^{=S}\|_{2}^{2}\le\sum_{T}\left(4d\right)^{\left|T\right|}\|L_{T}\left[f\right]\|_{4}^{4}.
\]

\subsection*{Upper bounding the contribution from $I_{2}$}

We now upper bound the contribution from $I_{2}.$ Let $T_{1}\Delta T_{2}=S.$
Then for each $S'\subsetneq S$, we assert that $A_{S'}\left(f^{=T_{1}}f^{=T_{2}}\right)=0$.
Let $i\in S\setminus S'$. Then $i\in T_{1}\Delta T_{2}.$ Assume
without loss of generality that $i\in T_{1}$. Then 
\[
A_{S'\cup T_{2}}\left(f^{=T_{1}}f^{=T_{2}}\right)=f^{=T_{2}}A_{S'\cup T_{2}}f^{=T_{1}}=0.
\]
 This shows that $A_{S'}\left[f^{=T_{1}}f^{=T_{2}}\right]=0.$ Hence,
\[
\left(f^{=T_{1}}f^{=T_{2}}\right)^{=S}=A_{S}\left(f^{=T_{1}}f^{=T_{2}}\right)=\left\langle f^{=T_{1}}\left(x,\cdot\right),f^{=T_{2}}\left(x,\cdot\right)\right\rangle _{L^{2}\left(\mu_{\overline{S}}\right)}.
\]

By Cauchy--Schwarz we have 
\begin{align*}
\|\sum_{\left(T_{1},T_{2}\right)\in I_{2}}\left(f^{=T_{1}}f^{=T_{2}}\right)^{=S}\|_{2}^{2} & =\sum_{T_{1}\Delta T_{2}=T_{3}\Delta T_{4}=S}\left\langle \left(f^{=T_{1}}f^{=T_{2}}\right)^{=S},\left(f^{=T_{3}}f^{=T_{4}}\right)^{=S}\right\rangle \\
 & \le\sum_{T_{1}\Delta T_{2}=T_{3}\Delta T_{4}=S}\left\Vert \left(f^{=T_{1}}f^{=T_{2}}\right)^{=S}\right\Vert _{2}\left\Vert \left(f^{=T_{3}}f^{=T_{4}}\right)^{=S}\right\Vert _{2}.
\end{align*}
 Now, for each $\left(T_{1},T_{2}\right)\in I_{2}$ we have 
\begin{align*}
\left\Vert \left(f^{=T_{1}}f^{=T_{2}}\right)^{=S}\right\Vert _{2}^{2} & =\mathbb{E}_{x\sim\mu_{S}}\left\langle f^{=T_{1}}\left(x,\cdot\right),f^{=T_{2}}\left(x,\cdot\right)\right\rangle _{L^{2}\left(\mu_{\overline{S}}\right)}^{2}\\
 & \le\mathbb{E}_{x\sim\mu_{S}}\left[\|f^{=T_{1}}\left(x,\cdot\right)\|_{L^{2}\left(\mu_{\overline{S}}\right)}^{2}\|f_{x}^{=T_{2}}\|_{L^{2}\left(\mu_{\overline{S}}\right)}^{2}\right]\\
 & =\mathbb{E}_{x\sim\mu_{S}}\|f_{x}^{=T_{1}}\|_{2}^{2}\mathbb{E}_{x\sim\mu_{S}}\|f_{x}^{=T_{2}}\|_{2}^{2}\\
 & =\|f^{=T_{1}}\|_{2}^{2}\|f^{=T_{2}}\|_{2}^{2},
\end{align*}
where in the second equality we used the fact that $\|f^{=T}\left(x,\cdot\right)\|_{L^{2}\left(\mu_{\overline{S}}\right)}^{2}$
depends only on $x_{T\cap S}$, so these are independent for $T=T_{1}$
and $T=T_{2}.$ This establishes
\[
\mathbb{E}\left[\left(f^{=T_{1}}f^{=T_{2}}\right)^{=S}\left(f^{=T_{3}}f^{=T_{4}}\right)^{=S}\right]\le\|f^{=T_{1}}\|_{2}\|f^{=T_{2}}\|_{2}\|f^{=T_{3}}\|_{2}\|f^{=T_{4}}\|_{2}
\]

Summing over all $S$, we obtain 
\begin{align*}
\sum_{S}\|\sum_{\left(T_{1},T_{2}\right)\in I_{2}}\left(f^{=T_{1}}f^{=T_{2}}\right)^{=S}\|_{2}^{2} & \le\mathbb{E}_{\left(\left\{ 0,1\right\} ^{n},\mu_{\frac{1}{2}}\right)}\left(\sum_{S\subseteq\left[n\right]}\|f^{=S}\|_{2}\chi_{S}\right)^{4}\\
 & \le9^{d}\mathbb{E}\left[(\sum\|f^{=S}\|_{2}\chi_{S})^{2}\right]^{2}\\
 & =9^{d}\|f\|_{2}^{4}.
\end{align*}
 Here the first inequality follows by expanding both terms and the
second is a well known consequence of hypercontractivity in the uniform
cube.

\subsection*{Showing that there is no contribution from $I_{3}$}

We recall that $I_{3}$ consist of the pairs with either $\left(T_{1}\Delta T_{2}\right)\setminus S\ne\varnothing$
or $S\setminus\left(T_{1}\cup T_{2}\right).$ Then we claim that $\left(f^{=T_{1}}f^{=T_{2}}\right)^{=S}=0$.
If $T_{1}\cup T_{2}$ does not contain $S$, then 
\[
f^{=T_{1}}f^{=T_{2}}=A_{T_{1}\cup T_{2}}\left(f^{=T_{1}}f^{=T_{2}}\right)=\sum_{S'\subseteq S}\left(f^{=T_{1}}f^{=T_{2}}\right)^{=S'}.
\]
 The uniqueness of the Efron--Stein decomposition shows that $\left(f^{=T_{1}}f^{=T_{2}}\right)^{=S}=0$.
Suppose now that there exists $i\in\left(T_{1}\Delta T_{2}\right)\setminus S.$
Without loss of generality $i\in T_{1}.$ We then have 
\[
A_{\left[k\right]\setminus\left\{ i\right\} }\left(f^{=T_{1}}f^{=T_{2}}\right)=f^{=T_{2}}\cdot A_{\left[k\right]\setminus\left\{ i\right\} }f^{=T_{1}}=0.
\]
 In particular, $\left(f^{=T_{1}}f^{=T_{2}}\right)^{=S}=0$ as for
each $S\subseteq\left[k\right]\setminus\left\{ i\right\} $ we have
\[
\left(f^{=T_{1}}f^{=T_{2}}\right)^{=S}\left[A_{\left[k\right]\setminus\left\{ i\right\} }\left(f^{=T_{1}}f^{=T_{2}}\right)\right]^{=S}=0
\]

\subsection*{Combining the contributions from $I_{1}$ and $I_{2}$.}

The lemma now follows by Cauchy--Schwarz. We have 
\begin{align*}
\|f\|_{4}^{4} & \le\sum_{S}\|\left(f^{2}\right)^{=S}\|_{2}^{2}\\
 & \le\sum_{S}\left(2\sum_{\left(T_{1},T_{2}\right)\in I_{1}}\|\left(f^{2}\right)^{=S}\|_{2}^{2}+2\sum_{\left(T_{1},T_{2}\right)\in I_{2}}\|\left(f^{2}\right)^{=S}\|_{2}^{2}\right)\\
 & \le2\sum_{T}\left(4d\right)^{\left|T\right|}\|L_{T}\left[f\right]\|_{4}^{4}+2\cdot9^{d}\|f\|_{2}^{4}.
\end{align*}
\end{proof}
Finally, using Lemma \ref{lem:Inductive approach}, we can derive
Theorem \ref{thm:Our version of hypercontrativity-product case} as
follows.
\begin{proof}[Proof of Theorem \ref{thm:Our version of hypercontrativity-product case}]

The proof is by induction on $d$. Since $D_{T,x}\left[f\right]$
is of degree $d-\left|T\right|$, we have 
\begin{align*}
\frac{1}{2}\|f\|_{4}^{4} & \le9{}^{d}\|f\|_{2}^{4}+\sum_{T\ne\varnothing}\left(4d\right)^{\left|T\right|}\|L_{T}\left[f\right]\|_{4}^{4}\\
 & =9^{d}\|f\|_{2}^{4}+\sum_{T\ne\varnothing}\left(4d\right)^{\left|T\right|}\mathbb{E}_{x\sim\mu_{T}}\|D_{T,x}\left[f\right]\|_{4}^{4}\\
 & \le9^{d}\|f\|_{2}^{4}+\sum_{T\ne\varnothing}2\cdot9^{d-\left|T\right|}\left(4d\right)^{\left|T\right|}\sum_{T'\subseteq\left[n\right]\setminus T}\left(8d\right)^{T'}\mathbb{E}_{x\sim\mu_{T\cup T'}}I_{T\cup T',x}^{2}\\
 & =9^{d}\|f\|_{2}^{4}+\sum_{T\cap T'=\varnothing}2^{\left|T'\right|+1}9^{d-\left|T\right|}\left(4d\right)^{\left|T\cup T'\right|}\mathbb{E}_{x\sim\mu_{T\cup T'}}\|D_{T'\cup T,x}\left[f\right]\|_{2}^{4}\\
 & \le9^{d}\sum_{T\subseteq S}\left(9d\right)^{\left|T\right|}\mathbb{E}_{x\sim\mu_{T}}I_{T,x}^{2}.
\end{align*}
\end{proof}

\section{\label{sec:eps-spaces}$\epsilon$-product spaces and the operators
$A_{S,T}$}

In this section, we present the framework of $\epsilon$-product spaces,
of which high dimensional expanders are a special case. We also define
key operators in this setting and show some basic properties that
they satisfy.

\subsection{Complexes having $\epsilon$-pseudorandom links.}

It is useful for us to consider measures on $V_{1}\times\cdots\times V_{k}$
rather than pure $(k-1)$-dimensional complexes, which can be identified
with subsets $S\subseteq V^{k}.$ Instead we identify a set with the
uniform measure over it.

\subsection*{Projected complexes}

Let $\mu$ be a probability measure on $V_{1}\times\cdots\times V_{k}.$
We say that $\mu$ is a, weighted $k$-partite, $(k-1)$-dimensional
complex. Let $S\subseteq\left[k\right]$ we write $\mu_{S}$ for the
projection of $\mu$ on $S$. We write $\mu_{i}$ rather than $\mu_{\left\{ i\right\} }.$
We write $V_{S}$ for the support of $\mu_{S}$ inside $\prod_{i\in S}V_{i}.$
We write $\overline{S}$ for the complement of $S$.

\subsection*{Restricted complexes}

Let $x\in V_{S}.$ We write $\mu_{x}$ for the measure on $V_{\overline{S}}$
given by 
\[
\mu_{x}\left(y\right)=\frac{\mu\left(x,y\right)}{\mu_{S}\left(x\right)}.
\]
 We write $V_{x}$ for the support of $\mu_{x}.$ We refer to $\left(V_{x},\mu_{x}\right)$
as the \emph{link} of $\mu$ on $x$.

\subsection*{$\epsilon$-pseudorandom weighted graphs}

Let $V_{1},V_{2}$ be finite sets. A measure $\mu$ on $V_{1}\times V_{2}$
can be thought of as a weighted bipartite graph. We say that $\mu$
is $\epsilon$-pseudorandom if for each $f_{1}\colon V_{1}\to\mathbb{R},$
$f_{2}\colon V_{2}\to\mathbb{R}$ we have 
\begin{align*}
\left|\mathbb{E}_{\left(x_{1},x_{2}\right)\sim\mu}\left[f_{1}\left(x_{1}\right)f_{2}\left(x_{2}\right)\right]-\mathbb{E}_{x_{1}\sim\mu_{1}}\left[f_{1}\left(x_{1}\right)\right]\mathbb{E}_{x_{2}\sim\mu_{2}}\left[f_{2}\left(x_{2}\right)\right]\right| & \le\\
\epsilon\sqrt{\mathrm{Var}_{x_{1}\sim\mu_{1}}\left[f_{1}\left(x_{1}\right)\right]\mathrm{Var}_{x_{2}\sim\mu_{2}}\left[f_{2}\left(x_{2}\right)\right]}.
\end{align*}
We let $A_{12}$ be the operator from $L^{2}\left(V_{1},\mu_{1}\right)$
to $L^{2}\left(V_{2},\mu_{2}\right)$ given by 
\[
A_{12}f\left(x\right)=\mathbb{E}_{y\sim\mu_{x}}\left[f\left(y\right)\right].
\]
We have the following standard lemma.
\begin{lem}
\label{lem:spectral graph theory } The following are equivalent.
\begin{enumerate}
\item $\mu$ is $\epsilon$-pseudorandom
\item $\|A_{12}-\mathbb{E}\|_{2\to2}\le\epsilon.$
\item The second eigenvalue of $\mathrm{A}_{12}^{*}\mathrm{A}_{12}$ is
$\le\epsilon^{2}.$
\end{enumerate}
\end{lem}

\subsection*{$\epsilon$-pseudorandom links}

Now let $\mu$ on $V_{1}\times\cdots\times V_{k}.$ We say that $\mu$
has\emph{ $\epsilon$-pseudorandom skeletons} if for each $S$ of
size $2$ the measure $\mu_{S}$ is $\epsilon$-pseudorandom.

We say that $\mu$ is \emph{$\epsilon$-product }if for each $S\subseteq\left[k\right]$
of size $\le k-2$ and each $x\in V_{S}$ the link $\mu_{x}$ has
$\epsilon$-pseudorandom skeletons.

In all that follows we assume that $\mu$ is an $\epsilon$-product
measure on $V_{1}\times\cdots\times V_{k}.$

\subsection*{Inheritance}

The definition of $\epsilon$-product makes it easy for inductive
type argument for the following reason.
\begin{lem}
Let $\mu$ on $\prod_{i=1}^{k}V_{i}$ be $\epsilon$-product. Let
$S,T\subseteq\left[k\right]$ be disjoint. Then for each $x\in V_{S}$,
the probability measure $\left(\mu_{x}\right)_{T}=\left(\mu_{S\cup T}\right)_{x}$
is $\epsilon$-product.
\end{lem}

\begin{proof}
All the skeletons of links of $\left(\mu_{S\to x}\right)_{T}$ are
also skeletons of links of $\mu.$
\end{proof}

\subsection*{Pseudorandomness as a measure of independence}

Let $S,T\subseteq\left[n\right]$. Then we have an operator $A_{S,T}:L^{2}\left(V_{S},\mu_{S}\right)\to L^{2}\left(V_{T},\mu_{T}\right).$
The operator is given by 
\[
A_{S,T}f\left(y\right)=\mathbb{E}_{x\sim\mu}\left[f\left(x_{S}\right)|x_{T}=y\right].
\]
 We write $A_{S,T}^{\mu}$ to stress that the operator is taken with
respect to $\mu.$ We write $A_{S}$ for $A_{\left[k\right],S},$
the operator given by restricting $S$ and taking expectation.

When $S,T$ are disjoint we expect $A_{S,T}f$ to be close to $\mathbb{E}\left[f\right],$
as in the product case $A_{S,T}$ is equal to the expectation. In
fact, we do have the following.
\begin{lem}
\label{lem:composing disjoint averaging operators } Let $\mu$ be
$\epsilon$-product. Let $S,T\subseteq\left[k\right]$ be disjoint,
and let $f\in L^{2}\left(V_{S},\mu_{S}\right)$. We have 
\[
\|A_{S,T}f-\mathbb{E}\left[f\right]\|_{2}^{2}\le\left|S\right|\left|T\right|\epsilon^{2}\|f\|_{2}^{2}
\]
\end{lem}

\begin{proof}
We prove it by induction on $k$. The case where $k=2$ is Lemma \ref{lem:spectral graph theory },
so we assume $k>2$. Given a probability space $\left(\Omega,\mu\right)$
we write $1^{\bot}$ for the subspace of $L^{2}\left(\Omega,\mu\right)$
consisting of functions that are orthogonal to the constant function
$1.$ We write $\tilde{\|}A_{S,T}\tilde{\|}$ for the $L^{2}$ operator
norm of $A_{S,T}$ as an operator from $1^{\bot}$ to $1^{\bot}$.
I.e. 
\[
\tilde{\|}A_{S,T}\tilde{\|}=\max_{f\in1^{\bot}}\frac{\|A_{S,T}f\|_{2}}{\|f\|_{2}}.
\]
 Our goal is to show that 
\[
\tilde{\|}A_{S,T}\tilde{\|}\le\sqrt{\left|S\right|\left|T\right|}\epsilon.
\]

\subsection*{Discarding the trivial cases}

If $T=\varnothing$, then $A_{S,T}=\mathbb{E}$ and the result is
trivial. If $S\cup T\ne\left[k\right]$ the result follows 
by working with the space $\left(V_{S\cup T},\mu_{S\cup T}\right)$
rather then $\left(V_{\left[k\right]},\mu_{\left[k\right]}\right).$
We also have $\tilde{\|}A_{S,T}\tilde{\|}=\tilde{\|}A_{S,T}^{*}\tilde{\|}$
as $A_{S,T}1=1.$ As $A_{S,T}^{*}=A_{T,S}$ we may assume that $\left|T\right|\le\left|S\right|.$
As $k>2$ we may therefore assume that $\left|T\right|\ge2.$

\subsection*{Completing the proof in the case where $\left|T\right|>1$}

Assume without loss of generality that $1\in T$.

Let $f\in1^{\bot}$. Using the fact that the equality 
\[
\|X\|_{2}^{2}=\mathbb{E}\left[X\right]^{2}+\|X-\mathbb{E}\left[X\right]\|_{2}^{2}
\]
 holds for every random variable $X$ we have 
\begin{align*}
\|A_{S,T}f\|_{2}^{2} & =\mathbb{E}_{y\sim\mu_{T}}\mathbb{E}_{x\sim\mu_{y}}^{2}\left[f\left(x_{S}\right)\right]\\
 & =\mathbb{E}_{a\sim\mu_{1}}\|A_{S,T\setminus\left\{ 1\right\} }^{\mu_{a}}f\|_{2}^{2}\\
 & =\mathbb{E}_{a\sim\mu_{1}}\left[\mathbb{E}_{\mu_{a}}^{2}f+\|A_{S,T\setminus\left\{ 1\right\} }^{\mu_{a}}f-\mathbb{E}_{\mu_{a}}f\|_{2}^{2}\right].\\
 & =\mathbb{E}_{a\sim\mu_{1}}\left[A_{S,1}f^{2}\left(a\right)+\|A_{S,T\setminus\left\{ 1\right\} }^{\mu_{a}}f-\mathbb{E}_{\mu_{a}}f\|_{2}^{2}\right]
\end{align*}
 By induction we may upper bound the right hand side we have 
\begin{align*}
RHS & \le\mathbb{E}_{a\sim\mu_{1}}\left[A_{S,1}f^{2}\left(a\right)+\left|S\right|\left|T-1\right|\epsilon^{2}\|f\|_{L^{2}\left(\mu_{a}\right)}^{2}\right].\\
 & =\|A_{S,1}f\|_{2}^{2}+\left|S\right|\left|T-1\right|\epsilon^{2}\|f\|_{2}^{2}.\\
 & \le\left|S\right|+\left|S\right|\left|T-1\right|\epsilon^{2}\|f\|_{2}^{2}\\
 & =\left|S\right|\left|T\right|\epsilon^{2}\|f\|_{2}^{2}
\end{align*}
\end{proof}

\subsection*{Understanding the operators $A_{S,T}$ and their compositions}

We now deduce that we have a similar upper bound of the form 
\[
\|A_{S,T}-A_{S,S\cap T}\|_{2\to2}\le\sqrt{\left|S\right|\left|T\right|}\epsilon.
\]

\begin{cor}
\label{cor:AST close to AS SCAPT} Let $S,T\subseteq\left[k\right]$,
and let $f\in L^{2}\left(\mu_{S}\right)$. Then 
\[
\|A_{S,T}f-A_{S,S\cap T}f\|_{2}^{2}\le\left|S\right|\left|T\right|\epsilon^{2}\|f\|_{2}^{2}.
\]
\end{cor}

\begin{proof}
Lemma \ref{lem:composing disjoint averaging operators } covers the
case $S\cap T=\varnothing.$ This shows that the corollary is true
in $\mu_{x}$ for each $x\in V_{S\cap T}.$ Therefore
\begin{align*}
\|A_{S,T}f-A_{S,S\cap T}f\|_{2}^{2} & =\mathbb{E}_{x\sim\mu_{S\cap T}}\|A_{S\setminus T,T\setminus S}^{\mu_{x}}f-A_{S\setminus T,\mathrm{\varnothing}}^{\mu_{x}}f\|_{L^{2}\left(\mu_{x}\right)}^{2}\\
 & \le\left|S\right|\left|T\right|\epsilon^{2}\mathbb{E}_{x}\|f\|_{L^{2}\left(\mu_{x}\right)}^{2}\\
 & =\left|S\right|\left|T\right|\epsilon^{2}\|f\|_{2}^{2}.
\end{align*}
\end{proof}
We now show that compositions behave similarly to the product space
setting.
\begin{lem}
\label{lem:composing averaging operators} We have

\[
\|A_{T_{2}}A_{T_{1}}-A_{T_{1}\cap T_{2}}\|_{2\to2}\le\left|T_{1}\right|\left|T_{2}\right|\epsilon.
\]
\end{lem}

\begin{proof}
We may assume that $T_{1}\cap T_{2}=\varnothing$. Indeed, if the
lemma holds for $T_{1}\cap T_{2}=\varnothing$ then it holds in general.
Indeed, write 
\[
\tilde{T_{1}}=T_{1}\setminus T_{2},\tilde{T_{2}}=T_{2}\setminus T_{1},A=\left[k\right]\setminus\left(T_{1}\cap T_{2}\right).
\]
 Let $x\in V_{T_{1}\cap T_{2}}$. Then we have
\[
\left(A_{T_{2}}A_{T_{1}}f\right)\left(x,\cdot\right)=\left(A_{\tilde{T}_{2}}^{\mu_{x}}A_{\tilde{T_{1}}}^{\mu_{x}}\right)\left(f\left(x,\cdot\right)\right)
\]
 and 
\[
A_{T_{1}\cap T_{2}}f\left(x,\cdot\right)=\mathbb{E}_{y\sim\mu_{x}}\left[f\left(x,y\right)\right].
\]
 Therefore once we prove the case $T_{1}\cap T_{2}=\varnothing$ it
would imply that for each $x$ 
\[
\mathbb{E}_{y\sim\mu_{x}}\left(A_{T_{2}}A_{T_{1}}f\left(x,y\right)-A_{T_{1}\cap T_{2}}f\left(x,y\right)\right)^{2}\le\left|T_{1}\right|\left|T_{2}\right|\epsilon^{2}\mathbb{E}_{y\sim\mu_{x}}f\left(x,y\right)^{2}.
\]
 The lemma will then follow by taking expectations over $x$.

Let us now settle the case $T_{1}\cap T_{2}=\varnothing$. Write $\mathrm{T}=A_{T_{2}}A_{T_{1}}.$
Then 
\[
\mathrm{T}=A_{T_{1},T_{2}}A_{T_{1}}.
\]
 Write $g=A_{T_{1}}f$. We have $\|g\|_{2}\le\|f\|_{2}$ by Cauchy--Schwarz.
By Lemma \ref{lem:composing disjoint averaging operators } we have
\begin{align*}
\|\mathrm{T}f-\mathbb{E}\left[f\right]\|_{2}^{2} & =\|A_{T_{1},T_{2}}g-\mathbb{E}g\|_{2}^{2}\\
 & \le\left|T_{1}\right|\left|T_{2}\right|\epsilon^{2}\|g\|_{2}^{2}\\
 & \le\left|T_{1}\right|\left|T_{2}\right|\epsilon^{2}\|f\|_{2}^{2}.
\end{align*}
\end{proof}

\section{\label{sec:Efron=002013Stein} Efron--Stein decompositions for link
expanders}

In this section, we introduce a new approximate Efron--Stein decomposition
for high dimensional expanders. In fact, it is more convenient to
state and prove our results in the more general setting of $\epsilon$-product
spaces, of which high dimensional expanders are a special case. We
proceed to discuss this setting below.

We first define the Efron--Stein decomposition via the usual formula
for it.
\begin{defn}
Let $f\in L^{2}\left(V,\mu\right)$ and $S\subseteq\left[n\right].$
We write 
\[
f^{=S}=\sum_{T\subseteq\left[S\right]}\left(-1\right)^{\left|S\setminus T\right|}A_{T}f.
\]
\end{defn}

The functions $f^{=S}$ are defined in terms of the operators $A_{T}$.
$L^{2}$-wise the composition of the operators $\left\{ A_{T}\right\} _{T\subseteq\left[k\right]}$
behave similarly to the compositions in the product case setting.
We satrt this section by making use of that and showing that many known
facts from the product setting generalize to the $\epsilon$-product
setting up to a small error.

\subsection{$L^{2}$-approximations for the Efron--Stein decomposition}

Thinking of $\epsilon$ as tending to $0$ in a much quicker pace
than $\frac{1}{k}$. Our goal is now to show that if $\mu$ is $\epsilon$-product,
then we have:
\begin{enumerate}
\item 
\[
\left|\|f\|_{2}^{2}-\sum_{S\subseteq\left[k\right]}\|f^{=S}\|_{2}^{2}\right|=o\left(\|f\|_{2}^{2}\right),
\]
\item and more generally
\[
\left|\left\langle f,g\right\rangle -\sum_{S}\left\langle f^{=S},g^{=S}\right\rangle \right|=o\left(\|f\|_{2}\|g\|_{2}\right).
\]
\end{enumerate}
One main tool involves the notion of a junta. We say that $g\colon V\to\mathbb{R}$
is a \emph{$T$-junta} if $g\left(x\right)$ depends only on $x_{T}.$
Equivalently, $g$ is a $T$-junta if $A_{T}g=g.$

Our first step towards the proof is a near orthogonality result between
$f^{=T}$ and $g^{=S}$ for $T\ne S.$

We start by a Fourier formula that holds exactly, this is unlike most
of the results in this section that only generalize the situation
from the product space setting up to a small error term.
\begin{lem}
\label{lem:ES Fourier formula for E_A} We have 
\[
A_{S}\left[f\right]=\sum_{T\subseteq S}f^{=T}\left(x\right).
\]
 In particular $f=\sum_{S\subseteq\left[k\right]}f^{=S}.$
\end{lem}

\begin{proof}
We have 
\begin{align*}
\sum_{T\subseteq S}f^{=T} & =\sum_{T\subseteq S}\sum_{T'\subseteq T}\left(-1\right)^{\left|T\setminus T'\right|}A_{T'}f\\
 & =\sum_{T'\subseteq S}A_{T'}f\sum_{T'\subseteq T\subseteq S}\left(-1\right)^{\left|T\setminus T'\right|}\\
 & =A_{S}f,
\end{align*}
 where the last equality follows from the fact that whenever $T'\ne S$
and $i\in S\setminus T'$ the pairs 
\[
\left(T,T\Delta\left\{ i\right\} \right)
\]
 contribute opposing signs to the sum $\sum_{T'\subseteq T\subseteq S}\left(-1\right)^{\left|T\setminus T'\right|}.$
The `in particular' part follows by taking $S=\left[k\right].$
\end{proof}
The following lemma holds even without assuming that $\mu$ is $\epsilon$-product.
\begin{lem}
\label{lem:Contraction properties } We have $\|A_{S,T}\|_{2\to2}\le1$
and 
\[
\|f^{=S}\|_{2}\le2^{\left|S\right|}\|f\|_{2}.
\]
\end{lem}

\begin{proof}
The triangle inequality implies that it suffices to prove the former
claim. Now by Cauchy--Schwarz we have 
\begin{align*}
\|A_{S,T}f\|_{2}^{2} & =\mathbb{E}_{x\sim\mu_{T}}A_{S,T}f\left(x\right)^{2}\\
 & =\mathbb{E}_{x\sim\mu_{T}}\left(\mathbb{E}_{y\sim\left(\mu_{x}\right)_{T}}f\left(y\right)\right)^{2}\\
 & \le\mathbb{E}_{x\sim\mu_{T}}\mathbb{E}_{y\sim\left(\mu_{x}\right)_{T}}f\left(y\right)^{2}\\
 & =\|f\|_{2}^{2}.
\end{align*}
\end{proof}
\begin{lem}
\label{lem:fS is othogonal to T-juntas} Let $f\colon V\to\mathbb{R}$,
 $T$ be a set not containing $S$, and $g$ be a $T$-junta.
Then 
\[
\left\langle f^{=S},g\right\rangle \le\epsilon\sqrt{\left|S\right|\left|T\right|}2^{\left|S\right|}\|f\|_{2}\|g\|_{2}.
\]
\end{lem}

\begin{proof}
As $A_{T}$ is the dual to the inclusion operator $L^{2}\left(V_{T}\right)\to L^{2}\left(V_{\left[k\right]}\right)$
we have 
\[
\left\langle f^{=S},g\right\rangle =\left\langle A_{T}f^{=S},g\right\rangle .
\]
By Cauchy--Schwarz it is sufficient to show that 
\[
\|A_{T}f^{=S}\|_{2}\le\epsilon\left|S\right|\left|T\right|2^{\left|S\right|}\|f\|_{2}.
\]
 Now 
\[
A_{T}f^{=S}=\sum_{S'\subseteq S}\left(-1\right)^{\left|S\setminus S'\right|}A_{T}A_{S'}f.
\]
 Roughly speaking, we rely on Lemma \ref{lem:ES Fourier formula for E_A},
which says that $\|A_{T}A_{S'}-A_{T\cap S'}\|_{2\to2}$ is small together
with the fact that 
\begin{equation}
\sum_{S'\subseteq S}\left(-1\right)^{\left|S\setminus S'\right|}A_{T\cap S'}f=0.\label{eq:12}
\end{equation}
The equality follows by choosing an arbitrary $i\in S\setminus T$
and noting that the sets $\left(S',S'\Delta\left\{ i\right\} \right)_{S'\subseteq S}$
correspond to the same term $A_{\left[k\right],T\cap S'}$, while
appearing with opposite signs. This shows that we have
\[
A_{T}f^{=S}=\sum_{S'\subseteq S}\left(-1\right)^{\left|S\setminus S'\right|}\left(A_{T}A_{S'}f-A_{T\cap S'}f\right).
\]
 By Lemma \ref{lem:composing averaging operators} we have 
\begin{align*}
 & \|A_{T}A_{S'}f-A_{T\cap S'}f\|_{2}\le\sqrt{\left|T\right|\left|S\right|}\epsilon\|f\|_{2}\le\sqrt{\left|S\right|\left|T\right|}\epsilon\|f\|_{2}.
\end{align*}
 Hence,
\[
\|A_{T}f^{=S}\|_{2}\le\sqrt{\left|S\right|\left|T\right|}2^{\left|S\right|}\epsilon.
\]
\end{proof}

\subsection*{Proof of our near orhogonality result}
\begin{cor}
\label{cor:Efron Stein corresponds to near orthogonal functions }
Let $T\ne S.$ Then $\left\langle f^{=S},g^{=T}\right\rangle \le2^{2\left|S\right|+2\left|T\right|}\epsilon\|f\|_{2}\|g\|_{2}.$
\end{cor}

\begin{proof}
The function $g^{=T}$ is a $T$-junta. By Lemmas \ref{lem:fS is othogonal to T-juntas}
and \ref{lem:Contraction properties } we therefore have the following
chain of inequalities if $T$ does not contain $S$.
\[
\left\langle f^{=S},g^{=T}\right\rangle \le\epsilon\sqrt{\left|S\right|\left|T\right|}2^{\left|S\right|}\|f\|_{2}\|g^{=T}\|_{2}\le\epsilon2^{2\left|S\right|+2\left|T\right|}\|f\|_{2}\|g\|_{2}.
\]
 A similar chain of inequalities holds when $S$ does not contain
$T.$
\end{proof}

\subsection*{Parseval holds approximately for the Efron--Stein decomposition}

\begin{lem}
\label{lem:Our Parseval }We have 
\[
\left|\left\langle f,g\right\rangle -\sum_{S\subseteq\left[k\right]}\left\langle f^{=S},g^{=S}\right\rangle \right|\le2^{4k}\epsilon\|f\|_{2}\|g\|_{2}.
\]
 Moreover, if $f$ is a $T$-junta, then 
\[
\left|\left\langle f,g\right\rangle -\sum_{S\subseteq T}\left\langle f^{=S},g^{=S}\right\rangle \right|=2^{4\left|T\right|}\epsilon\|f\|_{2}\|g\|_{2}.
\]
\end{lem}

\begin{proof}
We have $\left\langle f,g\right\rangle =\sum_{S\subseteq\left[k\right]}\left\langle f^{=S},g^{=S}\right\rangle +\sum_{S\ne T}\left\langle f^{=S},g^{=T}\right\rangle .$
By corollary \ref{cor:Efron Stein corresponds to near orthogonal functions }
we have 
\[
\sum_{S\ne T}\left\langle f^{=S},g^{=T}\right\rangle \le2^{4k}\epsilon\|f\|_{2}\|g\|_{2}.
\]
 For the `moreover' part note that if $f$ is a $T$-junta, then 
\[
\left\langle f,g\right\rangle =\left\langle f,A_{T}g\right\rangle _{L^{2}\left(\mu_{T}\right)}.
\]
 We may then apply the first part of the lemma in $\mu_{T}$ noting
that $\left(A_{T}g\right)^{=T'}=g^{=T'}$ for each $T'\subseteq T$.
\end{proof}

\subsection*{$\left(f^{=S}\right)^{=S}$ is $L^{2}$-close to $f^{=S}.$}

In the product space setting we have $\left(f^{=S}\right)^{=T}=\begin{cases}
f^{=S} & T=S\\
0 & T\ne S
\end{cases}.$ Here we have the following instead:
\begin{lem}
\label{lem:=00003DS is nearly idempotent} Let $g=f^{=S}.$ Then:
\begin{enumerate}
\item If $S\ne T$, then 
\[
\|g^{=T}\|_{2}^{2}\le2^{8k}\epsilon^{2}\|f\|_{2}^{2}
\]
\item 
\[
\|g^{=S}-g\|_{2}^{2}\le2^{10k}\epsilon^{2}\|f\|_{2}^{2}.
\]
\end{enumerate}
\end{lem}

\begin{proof}
We have 
\[
g^{=T}=\sum_{T'\subseteq T,S'\subseteq S}\left(-1\right)^{\left|S\setminus S'\right|+\left|T\setminus T'\right|}A_{T'}A_{S'}f.
\]
 Write 
\[
h=\sum_{T'\subseteq T,S'\subseteq S}\left(-1\right)^{\left|S\setminus S'\right|+\left|T\setminus T'\right|}A_{T'\cap S'}f.
\]
By Lemma \ref{cor:AST close to AS SCAPT} we therefore have 
\[
\|h-g^{=T}\|_{2}\le2^{2k}\max_{T',S'}\|A_{T'}A_{S'}-A_{T'\cap S'}\|_{2\to2}\|f\|_{2}\le2^{4k}\epsilon\|f\|_{2}.
\]
 Now we claim that $h=0.$ Indeed, assume without loss of generality
that $T$ is not contained in $S$ and let $i\in T\setminus S.$ Then
the terms $A_{T'\cap S'}$ appears with opposing sums for the pairs
$T'$ and $T'\Delta\left\{ i\right\} $.

(2)-follows by the fact that 
\[
\|g^{=S}-g\|_{2}=\left\Vert \sum_{T\ne S}g^{=T}\right\Vert _{2}\le\sum_{T\ne S}\|g^{=T}\|_{2}\le2^{5k}\epsilon\|f\|_{2}.
\]
\end{proof}

\subsection{Approximate Efron-Stein decomposition}

Again think of $\epsilon$ as tending to 0 much more quickly than
$\frac{1}{k}$. We now define a notion of $\left(\alpha,\epsilon'\right)$-approximate
Efron--Stein decomposition. We show that a version of Lemma \ref{lem:Our Parseval }
still holds for these approximate Efron--Stein decompositions.

\subsection*{Motivation}

One reason that demonstrates our need for an approximate Efron--Stein
decomposition is as follows. Let $f^{\le d}=\sum_{\left|S\right|<d}f^{=S}$.
Then we do not have 
\begin{align*}
\left(f^{\le d}\right)^{=S} & =\begin{cases}
f^{=S} & \left|S\right|\le d\\
0 & \left|S\right|>d
\end{cases},\\
\end{align*}
 but we would nevertheless like to work with the decomposition $\left\{ f^{=S}\right\} _{\left|S\right|\le d}$
as an approximate Efron--Stein decomposition for $f$. We capture
that notion as follows.

\subsection*{Defining the $\left(\alpha,\epsilon'\right)$-approximate Efron--Stein
decomposition}
\begin{defn}
\label{def:Approximate Efron-Stein} We say that $\left\{ f_{S}\right\} _{S\subseteq\left[k\right]}$
is an $\left(\alpha,\epsilon'\right)$-approximate Efron--Stein decomposition
if
\end{defn}

\begin{enumerate}
\item $\|f\|_{2}\le\alpha.$
\item 
\[
\|f-\sum_{S}f_{S}\|_{2}<\epsilon',
\]
\item For each $S$ there exists $h_{S}$ with $\|h_{S}\|_{2}\le\alpha$
and 
\[
\|h_{S}^{=S}-f_{S}\|_{2}\le\epsilon'.
\]
\end{enumerate}
It turns out that we have an approximate Parseval theorem for every
approximate Efron--Stein decomposition.
\begin{lem}
\label{lem:Strong Parseval} Let $\alpha_{1},\alpha_{2},\epsilon_{1},\epsilon_{2}>0$.
Suppose that $f$ has an $\left(\alpha_{1},\epsilon_{1}\right)$-bounded
approximate Efron--Stein decomposition $\left\{ f_{S}\right\} $
and $g$ has an $\left(\alpha_{2},\epsilon_{2}\right)$-bounded Efron--Stein
decomposition $\left\{ g_{S}\right\} .$ Then 
\[
\left|\left\langle f,g\right\rangle -\sum_{S}\left\langle f_{S},g_{S}\right\rangle \right|\le2^{6k}\left(\epsilon_{1}\alpha_{2}+\epsilon_{2}\alpha_{1}+\epsilon\alpha_{1}\alpha_{2}\right).
\]
\end{lem}

\begin{proof}
For each $S\subseteq\left[k\right]$ let $\tilde{f}_{S},\tilde{g}_{S}$
be with $\|\tilde{f}_{S}\|_{2}\le\alpha_{1},\|\tilde{g}_{S}\|_{2}\le\alpha_{2}$
\[
\|\tilde{f}{}_{S}^{=S}-f_{S}\|_{2}\le\epsilon_{1},
\]
 and 
\[
\|\tilde{g}_{S}^{=S}-g_{S}\|_{2}\le\epsilon_{2}.
\]
Let 
\[
f'_{S}=\tilde{f}_{S}^{=S},g'_{S}=\tilde{g}_{S}^{=S},
\]
 
\[
f'=\sum_{S\subseteq\left[k\right]}f'_{S}
\]
 and 
\[
g'=\sum_{S\subseteq\left[k\right]}g'_{S}.
\]
 By Lemma \ref{cor:Efron Stein corresponds to near orthogonal functions }
we have 
\begin{align*}
\left\langle f',g'\right\rangle  & =\sum_{S}\left\langle f'_{S},g'_{S}\right\rangle +\sum_{S\ne T\subseteq\left[k\right]}\left\langle f_{S}',g_{T}'\right\rangle \\
 & =\sum_{S}\left\langle f_{S}',g_{S}'\right\rangle \pm2^{6k}\epsilon\alpha_{1}\alpha_{2}.
\end{align*}
 Now by Cauchy--Schwarz 
\begin{align*}
\left\langle f,g\right\rangle  & =\left\langle f',g'\right\rangle +\left\langle f',g-g'\right\rangle +\left\langle f-f',g\right\rangle \\
 & =\sum_{S}\left\langle f'_{S},g'_{S}\right\rangle \pm\left(2^{6k}\epsilon\alpha_{1}\alpha_{2}+\|f'\|_{2}\|g-g'\|_{2}+\|f-f'\|_{2}\|g\|_{2}\right)\\
 & =\sum_{S}\left\langle f'_{S},g'_{S}\right\rangle \pm\left(2^{6k}\epsilon\alpha_{1}\alpha_{2}+2^{2k}\alpha_{1}\epsilon_{2}+\epsilon_{1}\alpha_{2}\right),
\end{align*}
 where the last equality used 
\[
\|f'\|_{2}\le\sum\|f'_{S}\|_{2}\le2^{k+\left|S\right|}\alpha_{1}\le2^{2k}\alpha_{1},
\]
 which follows from Lemma \ref{lem:Contraction properties }.

To complete the proof we note that we similarly have 
\begin{align*}
\left\langle f'_{S},g'_{S}\right\rangle  & =\left\langle f_{S},g_{S}\right\rangle \pm\|f_{S}\|_{2}\|g_{S}-g_{S}'\|_{2}+\|f_{S}'-f_{S}\|_{2}\|g_{S}'\|_{2}\\
 & =\left\langle f_{S},g_{S}\right\rangle \pm\alpha\epsilon_{2}+2^{k}\epsilon_{1}\alpha_{2}.
\end{align*}
\end{proof}
The above approximate Efron--Stein decomposition works well when
we care about $L_{2}$-norms. We actually care about closeness in
higher norms specifically $4$-norms. Our strategy when wishing to
upper bound $\|f-f'\|_{4}$ is to use the inequality
\[
\|f-f'\|_{4}^{4}\le\|f-f'\|_{2}^{2}\left(\|f\|_{\infty}+\|f'\|_{\infty}\right).
\]
 Where we hope that the $L^{2}$-closeness is sufficient to overcome
the loss of using infinity norms. We would therefore like everything
to have a relatively small infinity norm.
\begin{defn}
We say that $\left\{ f_{S}\right\} $ is a $\left(\beta,\alpha,\epsilon'\right)$-\emph{bounded
approximate Efron-Stein} decomposition if it is an $\left(\alpha,\epsilon'\right)$-approximate
Efron--Stein decomposition and moreover for each $S$: 
\[
\|h_{S}^{=S}\|_{\infty},\|f_{S}\|_{\infty},\|f\|_{\infty}
\]
 are all $\le\beta.$ Here $h_{S}^{=S}$ is as in Definition \ref{def:Approximate Efron-Stein}.
\end{defn}

We now show that the different Efron--Stein decompositions of a function
$f$ are all close in $L_{4}.$
\begin{lem}
\label{lem:L_4 Closeness of different Efron--Stein decompositions}
Suppose that $\left\{ f_{S}\right\} ,\left\{ f'_{S}\right\} $ are
$\left(\beta,\alpha,\epsilon'\right)$-bounded approximate Efron--Stein
decompositions for $f.$ Then
\begin{enumerate}
\item 
\[
\|f_{S}-f'_{S}\|_{2}^{2}\le O_{k}\left(\epsilon'\right)^{2}+O_{k}\left(\epsilon\alpha^{2}\right),
\]
\item 
\[
\|f_{S}-f_{S}'\|_{4}^{4}\le O_{k}\left(\epsilon'^{2}\beta^{2}\right)+O_{k}\left(\epsilon\alpha^{2}\beta^{2}\right),
\]
\item 
\[
\|\sum_{S}\left(f_{S}-f'_{S}\right)\|_{4}^{4}\le O_{k}\left(\epsilon'^{2}\beta^{2}\right)+O_{k}\left(\epsilon^{2}\alpha^{2}\beta^{2}\right),
\]
\item and 
\[
\|f-\sum_{S\subseteq\left[k\right]}f_{S}\|_{4}^{4}\le O_{k}\left(\epsilon^{2}\beta^{2}\right)\left(\alpha^{2}+\|f\|_{2}^{2}\right).
\]
\end{enumerate}
\end{lem}

\begin{proof}
(3) is an immediate corollary of (2). (4) also follows immediately
from (3) by setting $f'_{S}=f^{=S}$ while applying it with $2^{k}\beta$
rather than $\beta$. Indeed, $\|f^{=S}\|_{\infty}\le2^{k}\|f\|_{\infty}\le2^{k}\beta$.
Therefore $\left\{ f^{=S}\right\} _{S\subseteq\left[k\right]}$ is
a $\left(2^{k}\beta,\alpha,0\right)$-approximate Efron--Stein decomposition
for $f$. (2) follows immediately from (1) as we have 
\[
\|f_{S}-f'_{S}\|_{4}^{4}\le\|f_{S}-f'_{S}\|_{2}^{2}\|f_{S}-f'_{S}\|_{\infty}^{2}
\]
 and $\|f_{S}-f'_{S}\|_{\infty}^{2}\le4\beta^{2}.$

We now prove (1).

\subsection*{Reducing to the case that $f'_{S}=f^{=S}$}

First we assert that we may assume that $f'_{S}=f^{=S}$ for each
$S.$ Indeed, $\left\{ f^{=S}\right\} $ is a $\left(\beta,\alpha,0\right)$-Efron--Stein
decomposition. By the triangle inequality we have 
\[
\|f_{S}-f'_{S}\|_{2}\le\|f_{S}-f^{=S}\|_{2}+\|f^{=S}-f'_{S}\|_{2},
\]
 which implies (by H\'{o}lder) that 
\[
\|f_{S}-f'_{S}\|_{2}^{2}\le2\|f_{S}-f^{=S}\|_{2}^{2}+2\|f^{=S}-f'_{S}\|_{2}^{2}.
\]

This shows that it is sufficient to prove the theorem when $\left\{ f_{S}\right\} =\left\{ f^{=S}\right\} $
and when $\left\{ f'_{S}\right\} =\left\{ f^{=S}\right\} .$ Without
loss of generality we may assume that $f'_{S}=f^{=S}.$

\subsection*{Reducing to the case that $f_{S}=h_{S}^{=S}$}

Let $h_{S}$ be with $\|h_{S}\|_{2}\le\alpha$ and $\|f_{S}-h_{S}^{=S}\|_{2}<\epsilon'.$
Setting $\tilde{f}_{S}=h_{S}^{=S}$ we obtain by the triangle inequality
that $\left\{ \tilde{f}_{S}\right\} _{S\subseteq\left[k\right]}$
is a $\left(\beta,\alpha,\left(2^{k}+1\right)\epsilon'\right)$-bounded
approximate Efron--Stein decomposition for $f$. We have 
\[
\|f_{S}-f^{=S}\|_{2}^{2}\le2\|\tilde{f_{S}}-f_{S}\|_{2}^{2}+2\|\tilde{f}_{S}-f^{=S}\|_{2}\le2\epsilon'+2\|\tilde{f}_{S}-f^{=S}\|_{2}.
\]
 Therefore it is sufficient to prove (1) when $f_{S}$ is replaced
by $\tilde{f}_{S}$.

\subsection*{Proving the lemma when $f_{S}=h_{S}^{=S}$ and $f'_{S}=f^{=S}$}

By Cauchy--Schwarz and Corollary \ref{cor:Efron Stein corresponds to near orthogonal functions }
we have:
\begin{align*}
\left\langle f_{S}-f^{=S},f\right\rangle  & =\sum_{T\subseteq\left[k\right]}\left\langle f_{S}-f^{=S},f^{=T}\right\rangle \\
 & =\left\langle f_{S}-f^{=S},f^{=S}\right\rangle +\sum_{T\ne S}\left\langle f^{=S}-h_{S}^{=S},f^{=T}\right\rangle \\
 & =\left\langle f_{S}-f^{=S},f^{=S}\right\rangle +O_{k}\left(\epsilon\alpha^{2}\right).
\end{align*}
Again by Corollary \ref{cor:Efron Stein corresponds to near orthogonal functions }
and Cauchy--Schwarz we have: 
\begin{align*}
\left\langle f_{S}-f^{=S},f\right\rangle  & =\left\langle f_{S}-f^{=S},\sum_{T}f_{T}\right\rangle +\left\langle f_{S}-f^{=S},f-\sum_{T}f_{T}\right\rangle \\
 & =\left\langle f_{S}-f^{=S},f_{S}\right\rangle +O_{k}\left(\epsilon\alpha^{2}\right)+\|f_{S}-f^{=S}\|_{2}\epsilon'.
\end{align*}
Rearranging we obtain, 
\begin{align*}
\|f_{S}-f^{=S}\|_{2}^{2} & \le O_{k}\left(\epsilon'\right)\left(\|f^{=S}-f_{S}\|_{2}\right)+O_{k}\left(\epsilon\alpha^{2}\right).
\end{align*}
This shows that 
\[
\|f_{S}-f^{=S}\|_{2}^{2}\le O_{k}\left(\epsilon'\right)^{2}+O_{k}\left(\epsilon\alpha^{2}\right).
\]
\end{proof}

\section{\label{sec:overrview}Proof overview}

Building on the framework we established in Section 4, we can now
give a proof overview for our hypercontactive inequality on high dimensional
expanders. Recall that in the setting of direct products, we first
prove a key lemma, (Lemma \ref{lem:Inductive approach}) and then
use it to derive the theorem via an inductive argument. We now give
a sketch of how to generalise this approach to the $\epsilon$-product
setting.

\subsection{Generalising Lemma \ref{lem:Inductive approach}}

Recall that we would like to show a lemma of the form

\[
\|f\|_{4}^{4}\le C^{d}\|f\|_{2}^{4}+\sum_{S}\left(4d\right)^{\left|S\right|}\|L_{S}\left[f\right]\|_{4}^{4}.
\]
 We instead show a similar lemma that holds up to a small error term
of $O_{k}\left(\epsilon\|f\|_{2}^{2}\|f\|_{\infty}^{2}\right)$:

\begin{equation}
\frac{1}{2}\|f^{\le d}\|_{4}^{4}\le9^{d}\|f^{\le d}\|_{2}^{4}+4\sum_{0<\left|T\right|\le d}\left(4d\right)^{\left|T\right|}\|L_{T}^{\le d}\left[f\right]\|_{4}^{4}+O_{k}\left(\epsilon\right)\|f\|_{2}^{2}\|f\|_{\infty}^{2}.\label{eq:inductive approach in the epsilon product setting}
\end{equation}

However, first note that we do not have a useful notion of a low degree
function. Instead we work with 
\[
f^{\le d}=\sum_{\left|S\right|\le d}f^{=S}.
\]

In turn, instead of $L_{S}\left[f\right]$ we have 
\[
L_{S}^{\le d}\left[f\right]=\sum_{T\supseteq S,\left|T\right|\le d}f^{=T}.
\]
 We show that when expanding 
\[
\left(\left(f^{\le d}\right)^{2}\right)^{=S}=\sum_{T_{1},T_{2}}\left(f^{=T_{1}}f^{=T_{2}}\right)^{=S},
\]
 there are three kinds of terms: (1) terms that vanish in the product
space setting, but here they do not; (2) terms with $T_{1}\cap T_{2}\cap S\ne\varnothing$;
and (3) terms with $T_{1}\Delta T_{2}=S.$

Our high-level approach is to show that the same proof as in the setting
of product spaces works up to an error term. We accomplish that by
expressing everything in terms of our operators $\left\{ A_{S}\right\} $,
and we then replace equalities that hold in the product space by $L_{2}$-approximation
of the form 
\[
\|A_{S}A_{T}-A_{S\cap T}\|_{2\to2}\le O_{k}\left(\epsilon\right).
\]
 At first glance, it might appear that this approach would not suffice,
as we eventually would like to upper bound $4$-norms of terms, or
$2$-norms of expressions involving the product of two functions such
as $\left(f^{=T_{1}}f^{=T_{2}}\right)^{S}.$ Nevertheless, we are
able to accomplish that via inequalities of the form 
\[
\|f\|_{4}^{4}\le\|f\|_{2}^{2}\|f\|_{\infty}^{2}.
\]
 We then use the fact that all our terms are bounded by $O_{k}\left(\|f\|_{\infty}\right)$,
and our $L_{2}$-approximations involve $\epsilon$, and therefore
beat the $O_{k}\left(1\right)$-terms. This allows us to generalise
Lemma \ref{lem:Inductive approach} and prove (\ref{eq:inductive approach in the epsilon product setting}).

\subsection{Applying induction}

After having an inequality of the form 
\[
\|f^{\le d}\|_{4}^{4}\le C^{d}\|f^{\le d}\|_{2}^{2}+\sum_{S}\left(4d\right)^{\left|S\right|}\|L_{S}^{\le d}\left[f\right]\|_{4}^{4},
\]
 we would like to use a similar idea to the one we used in the product
space setting; that is, restrict $S$ to some $x\in V_{S}$, and then
apply induction for the function $L_{S}^{\le d}\left[f\right]\left(x,\cdot\right).$
The problem is that the restricted function $L_{S}^{\le d}\left[f\right]\left(x,\cdot\right)$
is no longer of degree $d-\left|S\right|$, and hence we can no longer
use induction.

We overcome this problem by using the notion of our approximate Efron--Stein
decompositions. Namely, we show that $L_{S}^{\le d}$ has two different
approximate Efron--Stein decomposition. The first one is 
\[
\left\{ f^{=T}\right\} _{T\supseteq S,\left|T\right|\le d},
\]
and the other one replaces $f^{=T}$ by the function $f_{T}$
\[
\left(x,y\right)\mapsto\left(L_{S}\left[f\right]\left(x,\cdot\right)\right)^{=T\setminus S}\left(y\right).
\]
 We then obtain that $\sum_{\left|T\right|\supseteq S,\left|T\right|\le d}f_{T}\left(x,\cdot\right)$
is of the form $D_{S,x}^{\le d-\left|S\right|}$, which allows us
to use induction similarly as in the product space setting.

After applying induction we get the compositions of two derivatives,
and we are again able to translate them back to expressions of the
form $\mathbb{E}_{x\sim\mu_{S}}I_{S,x}^{2}$ by showing that $D_{S,x}D_{T,y}$
and $D_{S\cup T,\left(x,y\right)}$ are both approximate Efron--Stein
decompositions of the same expression.

The remaining step is to upper bound the influences. We achieve that
by generalising the inequality 
\[
\mathbb{E}_{x\sim\mu_{S}}\left[I_{S,x}^{2}\right]\le\delta\|L_{S}\left[f\right]\|_{2}^{2}
\]
 from the product space setting, where crucially, we obtain that without
upper bounding $\|I_{S,x}\|_{\infty}$.

\section{\label{sec:notions}Laplacians, influences, and globalness on $\epsilon$-measures}

In this section, we define the notions of laplacians, derivatives
and influences in the setting of $\epsilon$-measures, give bounded
approximated Efron--Stein decompositions related to the Laplacians,
define globalness, and show that it implies small influences.

\subsection{Defining the Laplacians, derivatives and influences}
\begin{defn}
We define the Laplacians via the formula 
\[
L_{i}\left[f\right]=f-A_{\left[k\right]\setminus\left\{ i\right\} }f.
\]
\end{defn}

\begin{lem}
We have

\[
L_{i}\left[f\right]=\sum_{S\ni i}f^{=S}.
\]
\end{lem}

\begin{proof}
This follows immediately from Lemma \ref{lem:ES Fourier formula for E_A},
which shows that 
\[
A_{\left[k\right]\setminus\left\{ i\right\} }\left[f\right]=\sum_{S\subseteq\left[k\right]\setminus\left\{ i\right\} }f^{=S}.
\]
\end{proof}
\begin{defn}
We define $L_{S}\left[f\right]=\sum_{T\supseteq S}f^{=T}.$ Alternatively,
\[
L_{S}\left[f\right]=\sum_{T\subseteq S}\left(-1\right)^{\left|T\right|}A_{\left[k\right]\setminus T}f.
\]
 Let $x\in V_{S}.$ We let $D_{S,x}=L_{S}\left[f\right]\left(x,\cdot\right)$,
i.e. the function in $L^{2}\left(V_{x},\mu_{x}\right)$ obtained by
plugging in $x$ in the $S$ coordinates. We let 
\[
I_{S,x}\left[f\right]=\|D_{S,x}\left[f\right]\|_{L^{2}\left(V_{x},\mu_{x}\right)}.
\]
\end{defn}

\subsection{Bounded approximated Efron--Stein decompositions related to the
Laplacians}
\begin{lem}
\label{lem:L_S has a natural approximate Efron Stein }. There exists
$C=O_{k}\left(1\right)$, such that $\left\{ f^{=T}\right\} _{T\supseteq S}$
is a $\left(C\|f\|_{\infty},C\|f\|_{2},0\right)$-bounded approximate
Efron--Stein decomposition for $L_{S}\left[f\right].$
\end{lem}

\begin{proof}
We have 
\[
\|f^{=S}\|_{\infty}\le\sum_{T\subseteq S}\|A_{T}f\|_{\infty}\le2^{\left|S\right|}\|f\|_{\infty}
\]
 and 
\[
\|L_{S}\left[f\right]\|_{\infty}\le\sum_{T\subseteq S}\|A_{\overline{T}}f\|_{\infty}\le2^{\left|S\right|}\|f\|_{\infty}.
\]
 The other properties are easy to verify.
\end{proof}
Let $f\in L^{2}\left(V_{\left[k\right]},\mu_{\left[k\right]}\right)$
and let $g_{T}\in L^{2}\left(V_{T},\mu_{T}\right)$ be given by 
\[
g_{T}\left(x\right)=I_{T,x}\left[f\right].
\]
 Then $g_{T}$ can be interpreted interms of the Laplacians and the
averaging operators as 
\[
g_{T}=A_{T}\left(L_{T}\left[f\right]^{2}\right).
\]
 Suppose that $\left\{ f_{S}\right\} _{S\subseteq\left[k\right]}$
is a $\left(\beta,\alpha,\epsilon'\right)$-bounded approximate Efron--Stein
decompositions for $f$ and set $\widetilde{L_{T}\left[f\right]}=\sum_{S\supseteq T}f_{S}$.
The following lemma essentially shows that the function $A_{T}\left(\widetilde{L_{T}\left[f\right]}^{2}\right)$
is a good $L_{2}$-approximation for the function $g_{T}$. This can
be interpreted by saying that the generalised influences could be
computed via any $\left(\beta,\alpha,\epsilon'\right)$-bounded approximate
Efron--Stein decomposition for $f$.
\begin{lem}
\label{lem:Many ways to compute influences } Let $\left\{ f_{S}\right\} $
and $\left\{ f_{S}'\right\} $ be $\left(\beta,\alpha,\epsilon'\right)$-bounded
Efron-Stein decompositions for $f.$ Then 
\[
\|A_{T}\left(\sum_{S\supseteq T}f_{S}\right)^{2}\|_{2}^{2}\le2\|A_{T}\left(\sum_{S'\supseteq T}f'_{S}\right)^{2}\|_{2}^{2}+O_{k}\left(\epsilon'^{2}\beta^{2}\right)+O_{k}\left(\epsilon\alpha^{2}\beta^{2}\right).
\]
\end{lem}

\begin{proof}
By Cauchy--Schwarz we have 
\[
\left(\sum_{S\supseteq T}f_{S}\right)^{2}\le2\left(\sum_{S\supseteq T}f_{S}'\right)^{2}+2\left(\sum_{S\supseteq T}f_{S}'-f_{S}\right)^{2}.
\]
 Therefore 
\[
\|A_{T}\left(\sum_{S\supseteq T}f_{S}\right)^{2}\|_{2}^{2}\le2\|A_{T}\left(\sum_{S'\supseteq T}f'_{S}\right)^{2}\|_{2}^{2}+2\|A_{T}\left(\sum_{S\supseteq T}f'_{S}-f_{S}\right)^{2}\|_{2}^{2}.
\]
 Now since $A_{T}$ contracts 2-norms (Lemma \ref{lem:Contraction properties }).
We have 
\[
2\|A_{T}\left(\sum_{S\supseteq T}f'_{S}-f_{S}\right)^{2}\|_{2}^{2}\le2\|\sum_{S\supseteq T}f_{S}'-f_{S}\|_{4}^{4}.
\]
 Lemma \ref{lem:L_4 Closeness of different Efron--Stein decompositions}
now completes the proof.
\end{proof}
We now show that Lemma \ref{lem:L_S has a natural approximate Efron Stein }
is a special case of a more general phenomenon. Whenever $\left\{ f_{S}\right\} _{S\subseteq\left[k\right]}$
is a $\left(\beta,\alpha,\epsilon'\right)$-bounded approximate Efron--Stein
decomposition for $f$, we obtain that $\left\{ f_{T}\right\} _{T\supseteq S}$
is a $\left(\tilde{\beta},\tilde{\alpha},\tilde{\epsilon}\right)$-bounded
approximate Efron--Stein decomposition for suitable values of $\tilde{\beta},\tilde{\alpha},\tilde{\epsilon}$.
We show the following.
\begin{lem}
There exists $C=O_{k}\left(1\right)$, such that the following holds.
Suppose that $\left\{ f_{T}\right\} _{T\subseteq\left[k\right]}$
is a $\left(\beta,\alpha,\epsilon'\right)$-Approximate Efron--Stein
decomposition for $f$. Then $\left\{ f_{T}\right\} _{T\supseteq S}$
is a $\left(C\beta,\alpha,C\left(\epsilon'+\alpha\sqrt{\epsilon}\right)\right)$-Approximate
Efron--Stein decomposition for $L_{S}\left[f\right].$
\end{lem}

\begin{proof}
The only requirements that are not automatically inherited from $f$
are the upper bounds on $\|L_{S}\left[f\right]\|_{\infty}$, and on
$\|L_{S}f-\sum_{T\supseteq S}f_{T}\|_{2}.$ The former inequality
follows from the inequality 
\[
\|L_{S}\left[f\right]\|_{\infty}\le2^{\left|S\right|}\|f\|_{\infty}\le2^{k}\beta.
\]
While the latter follows from Lemma \ref{lem:L_4 Closeness of different Efron--Stein decompositions}
and the triangle inequality:
\begin{align*}
\|L_{S}f-\sum_{T\supseteq S}f_{T}\|_{2} & =\|\sum_{T\supseteq S}\left(f^{=T}-f_{T}\right)\|_{2}.\\
 & \le\sum_{T\supseteq S}\|f^{=T}-f_{T}\|_{2}\\
 & \le O_{k}\left(\epsilon'\right)+O_{k}\left(\alpha\sqrt{\epsilon}\right).
\end{align*}
\end{proof}

\subsection{Low degree functions and truncations}
\begin{defn}
We define the \emph{low degree part} of $f$ by setting 
\[
f^{\le d}=\sum_{\left|S\right|\le d}f^{=S}
\]
we define the\emph{ low degree Laplacians} of $f$ by setting 
\[
L_{T}^{\le d}\left[f\right]=\sum_{S\supseteq T,\left|S\right|\le d}f^{=T}.
\]
\end{defn}

We now show that if $\left\{ f_{T}\right\} _{T\subseteq\left[k\right]}$
is a $\left(\beta,\alpha,\epsilon'\right)$bounded approxiate Efron--Stein
decomposition for $f$, then we may turn it into an Efron--Stein
decomposition for $f^{\le d}$ and $L_{S}^{\le d}\left[f\right]$
in the obvious way.
\begin{lem}
\label{lem:approximate efron stein f=00005Cle d and L_s=00005Bf =00005D^=00005Cle d}
There exists $C=O_{k}\left(1\right)$, such that the following holds.
Suppose that $\left\{ f_{T}\right\} _{T\subseteq\left[k\right]}$
is a $\left(\beta,\alpha,\epsilon'\right)$-Approximate Efron--Stein
decomposition for $f$. Then
\begin{enumerate}
\item The functions $\left\{ f_{T}\right\} _{\left|T\right|\le d}$ are
a $\left(C\beta,\alpha,C\epsilon+C\epsilon'\right)$-approximate Efron--Stein
decomposition for $f^{\le d}.$
\item the functions $\left\{ f_{T}\right\} _{\left|T\right|\supseteq S,\left|T\right|\le d}$
are a $\left(C\beta,\alpha,C\epsilon+C\epsilon'\right)$-approximate
Efron--Stein decomposition for $L_{S}^{\le d}\left[f\right].$
\end{enumerate}
\end{lem}

\begin{proof}
It is sufficient to prove (2) as (1) is the special case where $S=\varnothing.$
By Lemma \ref{lem:L_4 Closeness of different Efron--Stein decompositions},
we have
\[
\|f_{T}-f^{=T}\|_{2}\le O_{k}\left(\epsilon'\right)+O_{k}\left(\sqrt{\epsilon}\alpha\right).
\]
 Hence, by the triangle inequality we have 
\begin{align*}
\|L_{S}^{\le d}\left[f\right]-\sum_{T\supseteq S,\left|T\right|\le d}f_{T}\|_{2} & \le\sum_{T\supseteq S,\left|T\right|\le d}\|f^{=T}-f_{T}\|_{2}\\
 & =O_{k}\left(\epsilon'\right)+O_{k}\left(\sqrt{\epsilon}\alpha\right).
\end{align*}
Moreover,
\begin{align*}
\|L_{S}^{\le d}\left[f\right]\|_{\infty} & =\left\Vert \sum_{T\supseteq S,\left|T\right|\le d}f^{=T}\right\Vert _{\infty}\\
 & \le2^{k}\max_{S}\|f^{=S}\|_{\infty}\\
 & \le4^{k}\|f\|_{\infty}\le4^{k}\beta.
\end{align*}
\end{proof}

\subsection{Globalness}

Unlike the product space setting the two possible definitions of globalness
are not equivalent. It turns out to be more convenient to work with
the notion concerning the restrictions.
\begin{defn}
We say that $f$ is $\left(d,\delta\right)$-global if for each $\left|S\right|\le d$ and each $x\in V_S$ 
we have 
\[
\|f\left(x,\cdot\right)\|_{L^{2}\left(V_{x},\mu_{x}\right)}\le\delta.
\]
\end{defn}

\begin{claim}
\label{claim:Globalness implies that fS is bounded } If $f$ is $\left(d,\delta\right)$-global
and $\epsilon$ is sufficiently small, then for each $T$ of size
$\le d$ we have 
\[
\|f^{=T}\|_{\infty}\le2^{\left|T\right|}\delta.
\]
\begin{proof}
This follows from the triangle inequality once we show that $\|A_{T'}f\|_{\infty}\le\delta$
for each $T'\subseteq T$. Indeed, for each $x$ we have 
\[
A_{T'}f\left(x\right)=\mathbb{E}_{\left(V_{x},\mu_{x}\right)}f\left(x,\cdot\right)\le\|f\left(x,\cdot\right)\|_{L^{2}\left(V_{x},\mu_{x}\right)}\le\delta.
\]
\end{proof}
\end{claim}

\begin{lem}
\label{lem:f le d ES} Suppose that $f$ is $\left(d,\delta\right)$-global.
Then $\left\{ f^{=S}\right\} _{\left|S\right|\le d}$ is a $\left(k^{d}\delta,\|f\|_{2},0\right)$-bounded
Efron--Stein decomposition for $f^{\le d}.$
\end{lem}

\begin{proof}
We have $\|f^{=S}\|_{\infty}\le2^{d}\delta$ by Claim \ref{claim:Globalness implies that fS is bounded }.
We also have 
\[
\|f\|_{\infty}\le\sum\|f^{=S}\|_{\infty}\le k^{d}\delta.
\]
 The rest of the conditions hold automatically.
\end{proof}
\begin{defn}
We say that $f$ is of $\left(\beta,\alpha\right)$-degree $d$ if
$f=\sum_{\left|S\right|\le d}f_{S}$ and $f_{S}=h_{S}^{=S}$ where
$\|h_{S}\|_{2}\le\alpha$, $\|h_{S}^{=S}\|_{\infty}\le\beta$ and
$\|f\|_{2}\le\alpha,\|f\|_{\infty}\le\beta.$
\end{defn}

If $f=\sum_{\left|S\right|\le d}f_{S}$ is of $\left(\beta,\alpha\right)$-degree
$d$ as above, then $\left\{ f_{S}\right\} $ is one $\left(\beta,\alpha,0\right)$-bounded
Efron--Stein decomposition for $f.$ We now show that in this case
the canonical $\left\{ f^{=S}\right\} _{\left|S\right|\le d}$ is
also $\left(\beta',\alpha',\epsilon'\right)$-bounded Efron--Stein
decomposition for the right parameters.

\subsection{Other approximate Efron-Stein decompositions for $L_{T}\left[f\right],L_{T}^{\le d}\left[f\right]$}
\begin{defn}
We define the low degree derivatives for $T\subseteq\left[k\right]$
and $x\in V_{T}$ 
\[
D_{T,x}^{\le d}\colon L^{2}\left(\mu\right)\to L^{2}\left(V_{x},\mu_{x}\right)
\]
via 
\[
D_{T,x}^{\le d}\left[f\right]=L_{T}^{\le d}\left[f\right]\left(x,\cdot\right)
\]
The low degree influences for $T\subseteq\left[n\right]$ and $x\in V_{T}$
are defined by 
\[
I_{T,x}^{\le d}\left[f\right]=\|L_{T}^{\le d}\left[f\right]\left(x,\cdot\right)\|_{L^{2}\left(V_{x},\mu_{x}\right)}^{2}.
\]
\end{defn}

We now move on to the critical lemma for our inductive approach. In
the product space setting our inductive approach relied on the fact
that $D_{T,x}f$ is of degree $\le d-\left|T\right|$ whenever $f$
is of degree $d.$ Here we show that $L_{T}\left[f\right]$ has an
alternative $\left(\beta,\alpha,\epsilon\right)$-bounded approximate
Efron--Stein decompositions $\left\{ f_{S}\right\} _{S\supseteq T}$
that gives rise to a function 
\[
\widetilde{L_{T}^{\le d}\left[f\right]}=\sum_{S\supseteq T,\left|S\right|\le d}f^{=S}
\]
 with the property that for each $x$ $\widetilde{L_{T}^{\le d}\left[f\right]}\left(x,\cdot\right)$
is of degree $d-\left|T\right|$.
\begin{lem}
\label{lem:Other approximate efron stein decompositions for the Laplacians.}
Let $f_{S}\left(x,y\right)=D_{T,x}^{=S\setminus T}[f](y).$ Then for each
$f$:
\end{lem}

\begin{enumerate}
\item The set $\left\{ f_{S}\right\} _{S\supseteq T}$ is a $\left(C\|f\|_{\infty},C\|f\|_{2},C\epsilon\|f\|_{2}\right)$-bounded
approximate Efron--Stein decomposition for $L_{T}\left[f\right].$
\item The set $\left\{ f_{S}\right\} _{S\supseteq T,\left|S\right|\le d}$
is a $\left(C\|f\|_{\infty},C\|f\|_{2},C\epsilon\|f\|_{2}\right)$-bounded
approximate Efron--Stein decomposition for $L_{T}^{\le d}\left[f\right].$
\item If $T'\supseteq T$, then the set $\left\{ f_{S}\right\} _{S\supseteq T'}$
is a $\left(C\|f\|_{\infty},C\|f\|_{2},C\epsilon\|f\|_{2}\right)$-bounded
approximate Efron Stein decomposition for $L_{T'}\left[f\right].$
\item The set $\left\{ f_{S}\right\} _{S\supseteq T',\left|S\right|\le d}$
is a $\left(C\|f\|_{\infty},C\|f\|_{2},C\epsilon\|f\|_{2}\right)$-bounded
approximate Efron--Stein decomposition for $L_{T'}^{\le d}\left[f\right].$
\item If $f$ is $\left(d,\delta\right)$-global. Then $\left\{ f_{S}\right\} _{S\supseteq T',\left|S\right|\le d}$
is a $\left(C\delta,C\|f\|_{2},C\epsilon\|f\|_{2}\right)$-bounded
approximate Efron--Stein decomposition for $L_{T'}^{\le d}\left[f\right].$
\end{enumerate}
\begin{proof}
Due to Lemma \ref{lem:approximate efron stein f=00005Cle d and L_s=00005Bf =00005D^=00005Cle d}
(1) implies (2)-(4). By Lemma \ref{lem:Contraction properties } all
the operators $A_{I}$ contract $\infty$-norms. We therefore have
\[
\|f_{S}\|_{\infty}\le\max_{x}2^{\left|S\setminus T\right|}\|D_{T,x}f\|_{\infty}=2^{\left|S\setminus T\right|}\|L_{T}f\|_{\infty}\le2^{\left|S\right|}\|f\|_{\infty}.
\]
To complete the proof it is sufficient to show that 
\[
\|f_{S}-f^{=S}\|_{2}\le O_{k}\left(\epsilon\|f\|_{2}\right)
\]
 as this will also imply that 
\begin{align*}
\|\sum_{S\supseteq T}f_{S}-L_{T}\left[f\right]\|_{2} & =\|\sum_{S\supseteq T}f_{S}-\sum_{S\supseteq T}f^{=S}\|_{2}\\
 & =O_{k}\left(\epsilon\|f\|_{2}\right).
\end{align*}
 We have

\begin{align*}
f_{S} & =\sum_{S'\subseteq S}\left(-1\right)^{\left|S\setminus S'\right|}A_{S'}L_{T}f.\\
 & =\sum_{S'\subseteq S}\sum_{T'\subseteq T}\left(-1\right)^{\left|S\setminus S'\right|+\left|T'\right|}A_{S'}A_{[k]\setminus T'}f.
\end{align*}

Write 
\begin{align*}
h & =\sum_{S'\subseteq S}\sum_{T'\subseteq T}\left(-1\right)^{\left|S\setminus S'\right|+\left|T'\right|}A_{S'\setminus T'}f
\end{align*}

 Then by Lemma \ref{cor:AST close to AS SCAPT} we have 
\[
\|f_{S}-h\|_{2}\le O_{k}\left(\epsilon\right)\|f\|_{2}.
\]
We then observe that whenever $S'\not\supseteq T$ the inner sum corresponding to it
is $0.$ In this case there is some $i\in T\setminus S'$ and
$T',T'\cup\left\{ i\right\} $ appear with alternating signs and correspond
to the same term $A_{S'\setminus T'}.$ Therefore we have
\begin{align*}
h & = \sum_{T\subseteq S'\subseteq S}\sum_{T'\subseteq T}\left(-1\right)^{\left|S\setminus S'\right|+\left|T'\right|}A_{S'\setminus T'}f = \sum_{S''\subseteq S}\left(-1\right)^{\left|S\setminus S''\right|}A_{S''}f = f^{=S}.
\end{align*} 
This shows that $\|f_{S}-f^{=S}\|_{2}\le O_{k}\left(\epsilon\right)\|f\|_{2},$
which completes the proof.  
\end{proof}

\subsection{Globalness implies small influences}

In the product space setting we had $\|I_{T,x}\|_{\infty}\le4^{d}\delta^{2}$
and we used it via the inequality 
\begin{equation}
\mathbb{E}_{x\sim\mu_{T}}\left[I_{T,x}^{2}\right]\le\mathbb{E}_{x\sim\mu_{T}}\left[I_{T,x}\right]4^{d}\delta^{2}=4^{d}\delta^{2}\|L_{T}\left[f\right]\|_{2}^{2}.\label{eq:189}
\end{equation}
 See the proof of Corollary \ref{cor:low degree functions have small 4-norms}.
Here we find a convoluted way of proving an analogue of (\ref{eq:189})
without having any upper bound on $\|I_{T,x}\|_{\infty}$ at our disposal.
\begin{lem}
\label{lem: Globalness implies small influences} Suppose that $f\colon V_{\left[k\right]}\to\mathbb{R}$
is $\left(d,\delta\right)$-global, and let $\left|T\right|\le d$.
Then
\[
\mathbb{E}_{x\sim\mu_{T}}\left[\left(I_{T,x}\left[f\right]\right)^{2}\right]\le2^{d+1}\delta^{2}\mathbb{E}_{x\sim\mu_{T}}I_{T,x}+O_{k}\left(\epsilon^{2}\|f\|_{4}^{4}\right).
\]
\end{lem}

\begin{proof}
Write $g\left(x\right)=I_{T,x}\left[f\right].$ Then $g=A_{T}\left[\left(L_{T}\left[f\right]\right)^{2}\right].$
We would like to upper bound $\|g\|_{2}^{2}.$ We accomplish that
by upper bounding $\|g\|_{2}^{2}$ by $\mathbb{E}\left[gg''\right]$
for a function $g''$ with a small $\infty$-norm.

By Cauchy--Schwarz we have 
\[
L_{T}\left[f\right]^{2}\le2^{\left|T\right|}\sum_{T'\supseteq\overline{T}}A_{T'}\left[f\right]^{2}.
\]
 This shows that 
\begin{equation}
g\le2^{\left|T\right|}\sum_{T'\subseteq T}A_{T}\left[\left(A_{T'}f\right)^{2}\right],\label{eq:g leq}
\end{equation}
 on all $x$. Let us denote by $g'$ the right hand side of \ref{eq:g leq}.
Also let 
\[
g''=2^{\left|T\right|}\sum_{T'\subseteq T}A_{T\cap T'}\left[\left(A_{T'}f\right)^{2}\right].
\]
 Then in the product space setting the functions $g',g''$ would have
been equal. Here we have an $L^{2}-$approximation between them.
\begin{claim}
\label{claim:g' is close to g''}$\|g'-g''\|_{2}\le O_{k}\left(\epsilon\right)\|f\|_{4}^{2}.$
\end{claim}

\begin{proof}
As $\left(A_{T'}f\right)^{2}$ is a $T'$-junta we have 
\[
\left(A_{T'}f\right)^{2}=A_{T'}\left[\left(A_{T'}f\right)^{2}\right].
\]
By Lemma \ref{cor:AST close to AS SCAPT} we have 
\[
\|A_{T}A_{T'}-A_{T\cap T'}\|_{2\to2}\le O_{k}\left(\epsilon\right).
\]
 We therefore have 
\begin{align*}
\|A_{T}\left[\left(A_{T'}f\right)^{2}\right]-A_{T\cap T'}\left[\left(A_{T'}f\right)^{2}\right]\|_{2}^{2} & \le O_{k}\left(\epsilon^{2}\right)\|\left(A_{T'}f\right)^{2}\|_{2}^{2}\\
 & \le O_{k}\left(\epsilon^{2}\right)\|f\|_{4}^{4},
\end{align*}
 as $A_{T'}$ contracts $4$-norms. Therefore, 
\[
\|g'-g''\|_{2}\le O_{k}\left(\epsilon\right)\|f\|_{4}^{2}.
\]
\end{proof}
As $0\le g\le g'$ we have 
\[
\mathbb{E}\left[g^{2}\right]\le\mathbb{E}\left[g'g\right]\le\mathbb{E}\left[g''g\right]+\mathbb{E}\left[\left(g'-g''\right)g\right].
\]
 By Cauchy--Schwarz we have 
\[
\mathbb{E}\left[\left(g'-g''\right)g\right]\le\|g'-g''\|_{2}\|g\|_{2}\le O_{k}\left(\epsilon\right)\|f\|_{4}^{2}\|g\|_{2}.
\]
 Now either 
\begin{equation}
\mathbb{E}\left[g^{2}\right]\le2\mathbb{E}\left[g''g\right]\label{eq:case 1}
\end{equation}
 or 
\[
\mathbb{E}\left[g^{2}\right]\le2\mathbb{E}\left[\left(g'-g''\right)g\right]\le O_{k}\left(\epsilon\right)\|f\|_{4}^{2}\|g\|_{2}.
\]
 In the latter case we have 
\begin{equation}
\|g\|_{2}^{2}\le O_{k}\left(\epsilon^{2}\right)\|f\|_{4}^{4}.\label{eq:case 2}
\end{equation}
 after rearranging. We can now sum the upper bounds of (\ref{eq:case 1})
and (\ref{eq:case 2}) corresponding to each of the cases to obtain
the upper bound
\[
\mathbb{E}\left[g^{2}\right]\le2\mathbb{E}\left[g''g\right]+O_{k}\left(\epsilon^{2}\right)\|f\|_{4}^{4}.
\]
that is true in both cases. The following claim completes the proof.
\begin{claim}
\label{claim:g'' has small infinity norm} $\|g''\|_{\infty}\le2^{d}\delta^{2}.$
\end{claim}

\begin{proof}
By Cauchy--Schwarz we point-wise have $\left(A_{T'}f\right)^{2}\le A_{T'}\left(f^{2}\right).$
We therefore have 
\[
A_{T\cap T'}\left[\left(A_{T'}f\right)^{2}\right]\le A_{T\cap T'}A_{T'}\left(f^{2}\right)=A_{T\cap T'}\left(f^{2}\right)\le\delta^{2}.
\]
 This shows that $\|g''\|_{\infty}\le2^{d}\delta^{2}.$
\end{proof}
\end{proof}
The same proof works for the truncated influences.
\begin{lem}
\label{lem:truncated influences are small} Suppose that $f\colon V_{\left[k\right]}\to\mathbb{R}$
is $\left(d,\delta\right)$-global. Suppose additionally that $\epsilon\le\epsilon_{0}\left(k\right)$.
Then we have 
\[
\mathbb{E}_{x\sim\mu_{T}}\left[\left(I_{T,x}^{\le d}\left[f\right]\right)^{2}\right]\le2^{d+4}\delta^{2}\mathbb{E}_{x\sim\mu_{T}}\left[I_{T,x}^{\le d}\right]+O_{k}\left(\epsilon^{2}\|f\|_{\infty}^{2}\|f\|_{2}^{2}\right).
\]
\end{lem}

\begin{proof}
Write 
\[
g_{1}\left(x\right)=\|\left(D_{T,x}\left[f\right]\right)^{\le d-\left|T\right|}\|_{2}^{2}.
\]

We now proceed with the following steps.

\subsection*{Upper bounding $\mathbb{E}_{x\sim\mu_{T}}\left[\left(I_{T,x}^{\le d}\left[f\right]\right)^{2}\right]$
in terms of $g_{1}$}

By Lemma \ref{lem:Other approximate efron stein decompositions for the Laplacians.}
the functions 
\[
\left\{ D_{T,x}\left[f\right]^{=S}\right\} _{\left|S\right|\le d-\left|T\right|},
\]
 is an alternative ($O_{k}\|f\|_{\infty},O_{k}\|f\|_{2},O_{k}\left(\epsilon\|f\|_{2}\right)$)-bounded
approximate Efron--Stein decomposition for $L_{T}^{\le d}\left[f\right]$.
Therefore by Lemma \ref{lem:Many ways to compute influences } we
have 
\begin{align}
\mathbb{E}_{x\sim\mu_{T}}\left[I_{T,x}^{\le d}\left[f\right]^{2}\right] & \le2\mathbb{E}_{x\sim\mu_{T}}\left[g_{1}\left(x\right)^{2}\right]+O_{k}\left(\epsilon\|f\|_{2}^{2}\|f\|_{\infty}^{2}\right).\label{eq:24601}
\end{align}

\subsection*{Repeating the proof of Lemma \ref{lem:truncated influences are small}}

Now by Lemma \ref{lem:Strong Parseval} we have $g_{1}\left(x\right)\le2I_{T,x}$
for each $x$, provided that $\epsilon$ is sufficiently small. Write
$g_{2}\left(x\right)=I_{T,x}\left[f\right].$ Similarly to the proof
of Lemma \ref{lem:truncated influences are small} we let 
\[
g_{2}'=2^{\left|T\right|}\sum_{T'\subseteq T}A_{T}\left[\left(A_{T'}f\right)^{2}\right]
\]
 and let 
\[
g_{2}''=2^{\left|T\right|}\sum_{T'\subseteq T}A_{T\cap T'}\left[\left(A_{T'}f\right)^{2}\right].
\]
By Cauchy--Schwarz we have 
\begin{align*}
\|g_{1}\|_{2}^{2}\le2\mathbb{E}\left[g_{1}g_{2}\right] & \le2\mathbb{E}\left[g_{1}g_{2}'\right]\le2\mathbb{E}\left[g_{1}g_{2}''\right]+2\mathbb{E}\left[g_{1}\left(g_{2}'-g_{2}''\right)\right].
\end{align*}
Now either 
\[
\|g_{1}\|_{2}^{2}\le4\mathbb{E}\left[g_{2}''g_{1}\right],
\]
 which would imply 
\[
\|g_{1}\|_{2}^{2}\le4\mathbb{E}\left[g_{2}''g_{1}\right]\le2^{d+2}\delta\mathbb{E}\left[g_{1}\right]
\]
 by Claim \ref{claim:g'' has small infinity norm}, or 
\begin{align*}
\|g_{1}\|_{2}^{2} & \le4\mathbb{E}\left[g_{1}\left(g_{2}'-g_{2}''\right)\right]\le4\|g_{1}\|_{2}\|g_{2}'-g_{2}''\|_{2}
\end{align*}
 and rearranging, we obtain 
\[
\|g_{1}\|_{2}^{2}\le16\|g_{2}'-g_{2}''\|_{2}^{2}\le O_{k}\left(\epsilon^{2}\|f\|_{4}^{4}\right)
\]
 by Claim \ref{claim:g' is close to g''}. This shows that 
\begin{equation}
\|g_{1}\|_{2}^{2}\le2^{d+2}\delta\mathbb{E}\left[g_{1}\right]+O_{k}\left(\epsilon^{2}\|f\|_{4}^{4}\right).\label{eq:bound for (g_1)}
\end{equation}

\subsection*{Moving back from $g_{1}$ to $I_{T,x}^{\le d}$}

By Lemmas \ref{lem:Other approximate efron stein decompositions for the Laplacians.}
we have 
\[
\|\left(D_{T,x}\left[f\right]\right)^{\le d-\left|T\right|}-D_{T,x}^{\le d}\left[f\right]\|_{2}^{2}\le O_{k}\left(\epsilon^{2}\|f\|_{2}^{2}\right)
\]
 yielding 
\begin{align}
\mathbb{E}\left[g_{1}\right] & \le2\|D_{T,x}^{\le d}\left[f\right]\|_{2}^{2}+O_{k}\left(\epsilon^{2}\|f\|_{2}^{2}\right)\label{eq:final}\\
 & =2\mathbb{E}_{x\sim\mu_{T}}I_{T,x}^{\le d}\left[f\right]+O_{k}\left(\epsilon^{2}\right)\|f\|_{2}^{2}\nonumber 
\end{align}
 by the triangle inequality and Cauchy--Schwarz. By combining (\ref{eq:24601}),
(\ref{eq:bound for (g_1)}) with (\ref{eq:final}) we obtain 
\begin{align*}
\mathbb{E}_{x\sim\mu_{T}}\left[I_{T,x}^{\le d}\left[f\right]^{2}\right] & \le2\|g_{1}\|_{2}^{2}+O_{k}\left(\epsilon\|f\|_{2}^{2}\|f\|_{\infty}^{2}\right).\\
 & \le2^{d+3}\delta^{2}\mathbb{E}\left[g_{1}\right]+O_{k}\left(\epsilon^{2}\|f\|_{2}^{2}\|f\|_{\infty}^{2}\right)+O_{k}\left(\epsilon^{2}\|f\|_{4}^{4}\right)\\
 & \le2^{d+3}\delta^{2}\mathbb{E}\left[g_{1}\right]+O_{k}\left(\epsilon^{2}\|f\|_{2}^{2}\|f\|_{\infty}^{2}\right).
\end{align*}
The lemma now follows by putting everything together.
\end{proof}

\section{Proving hypercontractivity for $\epsilon$-product measures\label{sec:Proving-hypercontractivity}}

We suggest revisiting Section 2 before reading this section. Our strategy
is the same as in the product case, and we deal with the differences
by appealing to the tools developed in Sections 3-6.

\subsection{Upper bounding $\|f^{\le d}\|_{4}^{4}$ by $4$-norms of non-trivial
Laplacians and $\|f^{\le d}\|_{2}^{4}$}

We now move on to preparing the ground for the proof of our hypercontractive
inequality.
\begin{lem}
\label{lem:4-norms of Laplacians } Let $f$ be $\left(d,\delta\right)$-global.
Suppose that $\epsilon\le\epsilon_{0}\left(k\right)$. Then we have
\[
\frac{1}{2}\|f^{\le d}\|_{4}^{4}\le9^{d}\|f^{\le d}\|_{2}^{4}+4\sum_{0<\left|T\right|\le d}\left(4d\right)^{\left|T\right|}\|L_{T}^{\le d}\left[f\right]\|_{4}^{4}+O_{k}\left(\epsilon\right)\|f\|_{2}^{2}\|f\|_{\infty}^{2}.
\]
\end{lem}

\begin{proof}
Let $g=f^{\le d}.$ By Lemma \ref{lem:Strong Parseval} we have 
\[
\|g\|_{4}^{4}\le2\sum_{S}\|\left(g^{2}\right)^{=S}\|_{2}^{2}.
\]
We now upper bound $\|\left(g^{2}\right)^{=S}\|_{2}^{2}.$ We have
\[
\left(g^{2}\right)^{=S}=\sum_{\left|T_{1}\right|\le d,\left|T_{2}\right|\le d}\left(f^{=T_{1}}f^{=T_{2}}\right)^{=S}.
\]
 Let
\begin{enumerate}
\item $I_{1}=\left\{ \left(T_{1},T_{2}\right):T_{1}\cap T_{2}\cap S\ne\varnothing\right\} .$
\item $I_{2}=\left(\left(T_{1},T_{2}\right):T_{1}\Delta T_{2}=S\right)$
\item $I_{3}=\left(T_{1}\Delta T_{2}\right)\setminus S\ne\varnothing$ or
$S\setminus\left(T_{1}\cup T_{2}\right)\ne\varnothing$.
\end{enumerate}
Our first step is to show that the contribution from $I_{3}$ is negligible.
This is to be expected as in the product space setting we were able
to show that the contribution from $I_{3}$ is 0.
\begin{claim}
Let $\left(T_{1},T_{2}\right)\in I_{3}$. Then 
\[
\|\left(f^{=T_{1}}f^{=T_{2}}\right)^{=S}\|_{2}^{2}\le O_{k}\left(\epsilon^{2}\right)\|f\|_{2}^{2}\|f\|_{\infty}^{2}.
\]
\end{claim}

\begin{proof}
Suppose first that $\left(T_{1}\Delta T_{2}\right)\setminus S\ne\varnothing.$
Then without loss of generality we may assume that there is some $i\in T_{1}\setminus\left(T_{2}\cup S\right).$
By Lemma \ref{lem:Strong Parseval} we have 
\[
\|\left(f^{=T_{1}}f^{=T_{2}}\right)^{=S}\|_{2}^{2}\le2\|A_{\left[k\right]\setminus\left\{ i\right\} }\left(f^{=T_{1}}f^{=T_{2}}\right)\|_{2}^{2}.
\]
 Now 
\[
A_{\left[k\right]\setminus\left\{ i\right\} }\left(f^{=T_{1}}f^{=T_{2}}\right)=\left(A_{\left[k\right]\setminus\left\{ i\right\} }\left(f^{=T_{1}}\right)\right)f^{=T_{2}}.
\]
 By Lemma \ref{lem:=00003DS is nearly idempotent} we have 
\[
\|A_{\left[k\right]\setminus\left\{ i\right\} }f^{=T_{1}}\|_{2}\le\sum_{T'\not\ni i}\|\left(f^{=T_{1}}\right)^{=T'}\|_{2}\le O_{k}\left(\epsilon\right)\|f\|_{2}.
\]
This shows that
\begin{align*}
\|A_{\left[k\right]\setminus\left\{ i\right\} }\left(f^{=T_{1}}f^{=T_{2}}\right)\|_{2}^{2} & \le\|A_{\left[k\right]\setminus\left\{ i\right\} }f^{=T_{1}}\|_{2}^{2}\|f^{=T_{2}}\|_{\infty}^{2}\\
 & \le O_{k}\left(\epsilon^{2}\right)\|f\|_{\infty}^{2}\|f\|_{2}^{2}.
\end{align*}
 Suppose now that $S\setminus\left(T_{1}\cup T_{2}\right)\ne\varnothing.$
Let $i\in S\setminus\left(T_{1}\cup T_{2}\right).$ Then the function
$g = f^{=T_{1}}f^{=T_{2}}$ is a $T_{1}\cup T_{2}$-junta. This shows
that $g=A_{T_{1}\cup T_{2}}g.$

Hence by Lemma \ref{lem:=00003DS is nearly idempotent} and the triangle
inequality we have 
\begin{align*}
\|g^{=S}\|_{2} & =\|\left(A_{T_{1}\cup T_{2}}g\right)^{=S}\|_{2}\\
 & \le\|\sum_{T\subseteq T_{1}\cup T_{2}}\left(g^{=T}\right)^{=S}\|_{2}\\
 & \le O_{k}\left(\epsilon\right)\|g\|_{2}.
\end{align*}
 It now remains to note that $\|g\|_{2}^{2}\le\|f^{=T_{1}}\|_{2}\|f^{=T_{2}}\|_{\infty}\le2^{2k}\|f\|_{2}\|f\|_{\infty}.$
\end{proof}
We now move on to our next step of upper bounding the contribution
from the pairs in $I_{1}.$
\begin{claim}
$\sum_{\left(T_{1},T_{2}\right)\in I_{1}}\left(f^{=T_{1}}f^{=T_{2}}\right)^{=S}=\sum_{T\subseteq S}\left(-1\right)^{\left|T\right|+1}\|\left(L_{T}^{\le d}\left[f\right]^{2}\right)^{=S}\|_{2}^{2}.$
\end{claim}

\begin{proof}
The proof is exactly the same as in the product case so we omit it.
\end{proof}
It now remains to consider the contribution from $I_{2}$, i.e. the
case $T_{1}\Delta T_{2}=S$ . Here just like the product case it is
sufficient to show the following claim
\begin{claim}
Let $T_{1}\Delta T_{2}=S$. Then we have 
\[
\|\left(f^{=T_{1}}f^{=T_{2}}\right)^{=S}\|_{2}\le2\|f^{=T_{1}}\|_{2}\|f^{=T_{2}}\|_{2}+O_{k}\left(\epsilon\right)\|f\|_{2}\|f\|_{\infty},
\]
 provided that $\epsilon$ is sufficiently small.
\end{claim}

\begin{proof}
First let $S'\subsetneq S.$ As $T_{1}\Delta T_{2}=S$, there exists
$i\in\left(T_{1}\Delta T_{2}\right)\setminus S'.$ Without loss of
generality $i\in T_{1}.$ By Lemmas \ref{lem:Strong Parseval}, \ref{lem:=00003DS is nearly idempotent},
and \ref{lem:ES Fourier formula for E_A} we have 
\begin{align*}
\|A_{S'}\left(f^{=T_{1}}f^{=T_{2}}\right)\|_{2} & \le2\|A_{S'\cup T_{2}}\left(f^{=T_{1}}f^{=T_{2}}\right)\|_{2}\\
 & \le2\|A_{S'\cup T_{2}}f^{=T_{1}}\|_{2}\|f^{=T_{2}}\|_{\infty}\\
 & \le O_{k}\left(\epsilon\right)\|f\|_{2}\|f\|_{\infty}.
\end{align*}
 By the triangle inequality this shows that 
\begin{align*}
\|\left(f^{=T_{1}}f^{=T_{2}}\right)^{=S}\|_{2} & \le\sum_{S'}\left(-1\right)^{\left|S\setminus S'\right|}A_{S'}\left(f^{=T_{1}}f^{=T_{2}}\right).\\
 & \le\|A_{S}\left(f^{=T_{1}}f^{=T_{2}}\right)\|_{2}+O_{k}\left(\epsilon\right)\|f\|_{2}\|f\|_{\infty}
\end{align*}

We now upper bound $\|A_{S}\left(f^{=T_{1}}f^{=T_{2}}\right)\|_{2}.$

By Cauchy--Schwarz for $x\in V_{S}$ we have 
\begin{align*}
A_{S}\left(f^{=T_{1}}f^{=T_{2}}\right)\left(x\right) & =\left\langle f^{=T_{1}}\left(x,\cdot\right),f^{=T_{2}}\left(x,\cdot\right)\right\rangle _{L^{2}\left(V_{x},\mu_{x}\right)}\\
 & \le\|f^{=T_{1}}\left(x,\cdot\right)\|_{L^{2}\left(V_{x},\mu_{x}\right)}\|f^{=T_{2}}\left(x,\cdot\right)\|_{L^{2}\left(V_{x},\mu_{x}\right)}.
\end{align*}
 This shows that 
\begin{align}
\|A_{S}\left(f^{=T_{1}}f^{=T_{2}}\right)\|_{2}^{2} & \le\mathbb{E}_{x\sim\mu_{S}}\left[\|f^{=T_{1}}\left(x,\cdot\right)\|_{L^{2}\left(V_{x},\mu_{x}\right)}^{2}\|f^{=T_{2}}\left(x,\cdot\right)\|_{L^{2}\left(V_{x},\mu_{x}\right)}^{2}\right].\label{eq:scary}
\end{align}
 We have 
\[
\|f^{=T_{1}}\left(x,\cdot\right)\|_{L^{2}\left(V_{x},\mu_{x}\right)}^{2}=A_{S}\left[\left(f^{=T_{1}}\right)^{2}\right]=A_{S}A_{T_{1}}\left(f^{=T_{1}}\right)^{2}
\]
 By Lemma \ref{cor:AST close to AS SCAPT}
\[
\|A_{S}A_{T_{1}}-A_{S\cap T_{1}}\|_{2\to2}\le O_{k}\left(\epsilon\right).
\]
 Hence,
\begin{align*}
\|A_{S}A_{T_{1}}\left(f^{=T_{1}}\right)^{2}-A_{S\cap T_{1}}\left[\left(f^{=T_{1}}\right)^{2}\right]\|_{2}^{2} & \le O_{k}\left(\epsilon^{2}\|f^{=T_{1}}\|_{4}^{4}\right)\\
 & \le O_{k}\left(\epsilon^{2}\|f\|_{2}^{2}\|f\|_{\infty}^{2}\right).
\end{align*}
By Cauchy--Schwarz this shows that 
\begin{align*}
\text{RHS of (\ref{eq:scary})} & =\left\langle A_{S}A_{T_{1}}\left(f^{=T_{1}}\right)^{2},A_{S}\left(f^{=T_{2}}\right)^{2}\right\rangle _{L^{2}\left(V_{S},\mu_{S}\right)}\\
 & =\left\langle A_{S\cap T_{1}}\left(f^{=T_{1}}\right)^{2},A_{S}\left(f^{=T_{2}}\right)^{2}\right\rangle _{L^{2}\left(V_{S},\mu_{S}\right)}\\
 & +O_{k}\left(\epsilon\|f\|_{2}\|f\|_{\infty}\right)\|A_{S}\left(f^{=T_{2}}\right)^{2}\|_{2}.
\end{align*}
 As we have 
\[
\|A_{S}\left(f^{=T_{2}}\right)^{2}\|_{2}^{2}\le O_{k}\left(\|f\|_{4}^{4}\right)\le O_{k}\left(\|f\|_{2}^{2}\|f\|_{\infty}^{2}\right).
\]
 Therefore, 
\begin{align*}
\|A_{S}\left(f^{=T_{1}}f^{=T_{2}}\right)\|_{2}^{2} & \le\left\langle A_{S\cap T_{1}}\left(f^{=T_{1}}\right)^{2},A_{S}\left(f^{=T_{2}}\right)^{2}\right\rangle _{L^{2}\left(V_{S},\mu_{S}\right)}+O_{k}\left(\epsilon\|f\|_{2}^{2}\|f\|_{\infty}^{2}\right)\\
 & =\left\langle A_{S\cap T_{1}}\left(f^{=T_{1}}\right)^{2},A_{S\cap T_{1}}\left(f^{=T_{2}}\right)^{2}\right\rangle _{L^{2}\left(\mu\right)}+O_{k}\left(\epsilon\|f\|_{2}^{2}\|f\|_{\infty}^{2}\right).
\end{align*}
 Now
\[
A_{S\cap T_{1}}\left(f^{=T_{2}}\right)^{2}=A_{S\cap T_{1}}A_{T_{2}}\left(f^{=T_{2}}\right)^{2}
\]
 and $\|A_{S\cap T_{1}}A_{T_{2}}-\mathbb{E}\|_{2\to2}\le\epsilon$
by Lemma \ref{cor:AST close to AS SCAPT}. Therefore we similarly
have 
\begin{align*}
\left\langle A_{S\cap T_{1}}\left(f^{=T_{1}}\right)^{2},A_{S\cap T_{1}}\left(f^{=T_{2}}\right)^{2}\right\rangle _{L^{2}\left(\mu\right)} & =\left\langle A_{S\cap T_{1}}\left(f^{=T_{1}}\right)^{2},\|f^{=T_{2}}\|_{2}^{2}\right\rangle +O_{k}\left(\epsilon\|f\|_{2}^{2}\|f\|_{\infty}^{2}\right).\\
 & =\|f^{=T_{1}}\|_{2}^{2}\|f^{=T_{2}}\|_{2}^{2}+O_{k}\left(\epsilon\|f\|_{2}^{2}\|f\|_{\infty}^{2}\right).
\end{align*}
 This completes the proof of the claim.
\end{proof}
The rest of the proof is the exactly the same as in the product case
setting. 
\end{proof}
Now the only thing to remains is to apply the inductive hypothesis.
\begin{thm}
We have $\|f^{\le d}\|_{4}^{4}\le20^{d}\sum_{\left|S\right|\le d}\left(4d\right)^{\left|S\right|}\mathbb{E}_{x\sim\mu_{S}}I_{S,x}^{\le d}\left[f\right]^{2}+O_{k}\left(\epsilon^{2}\right)\|f\|_{2}^{2}\|f\|_{\infty}^{2}.$
\end{thm}

\begin{proof}
The proof is by induction on $d$. By Lemma \ref{lem:4-norms of Laplacians }
we have 
\begin{equation}
\|f{}^{\le d}\|_{4}^{4}\le2\cdot9^{d}\|f^{\le d}\|_{2}^{4}+2\cdot\sum_{S\ne\varnothing}\left(4d\right)^{\left|S\right|}\|L_{S}^{\le d}\left[f\right]\|_{4}^{4}+O_{k}\left(\epsilon\right)\|f\|_{2}^{2}\|f\|_{\infty}^{2}.\label{eq:12-1}
\end{equation}
 Write $g_{S,x}\left(y\right)=\left(D_{S,x}\left[f\right]\right)^{\le d-\left|T\right|}\left(y\right).$
Then by Lemma \ref{lem:Other approximate efron stein decompositions for the Laplacians.}
and Lemma \ref{lem:L_4 Closeness of different Efron--Stein decompositions}
we have:
\[
\mathbb{E}_{x}\|D_{S,x}^{\le d}\left[f\right]\|_{4}^{4}\le2\mathbb{E}_{x}\|g_{S,x}\|_{4}^{4}+O_{k}\left(\epsilon^{2}\right)\|f\|_{2}^{2}\|f\|_{\infty}^{2}.
\]
 By induction, we have 
\begin{align}
\|g_{S,x}\|_{4}^{4} & \le20^{d-\left|S\right|}\sum_{T\cap S=\varnothing,\left|T\right|\le d-\left|S\right|}\left(4d\right)^{\left|T\right|}\mathbb{E}_{y\sim\mu_{T}}I_{T,y}^{2}\left[g_{S,x}\right]+O_{k}\left(\epsilon\right)\|g_{S,x}\|_{2}^{2}\|g_{S,x}\|_{\infty}^{2}.\label{eq:scary2}
\end{align}
 By Lemma \ref{lem:Other approximate efron stein decompositions for the Laplacians.}
we have $\|g_{S,x}\|_{\infty}=O_{k}\left(\|f\|_{\infty}\right)$.
By Lemmas \ref{lem:Other approximate efron stein decompositions for the Laplacians.},
\ref{lem:L_4 Closeness of different Efron--Stein decompositions},
and \ref{lem:Strong Parseval} we have 
\begin{align}
\mathbb{E}_{x\sim\mu_{S}}\|g_{S,x}\|_{2}^{2} & \le2\mathbb{E}_{x\sim\mu_{S}}\|D_{S,x}^{\le d}f\|_{2}^{2}+O_{k}\left(\epsilon^{2}\right)\|f\|_{2}^{2}\label{eq:scary3}\\
 & \le O_{k}\left(\|f\|_{2}^{2}\right)\nonumber 
\end{align}

Taking expectations over (\ref{eq:scary2}), and plugging in (\ref{eq:scary3})
we obtain:

\[
\mathbb{E}_{x}\|g_{S,x}\|_{4}^{4}\le2\sum_{T\cap S=\varnothing,\left|T\right|\le d-\left|S\right|}\left(4d\right)^{\left|T\right|}\mathbb{E}_{\left(x,y\right)\sim\mu_{S\cup T}}I_{T,y}^{2}\left[g_{S,x}\right]+O_{k}\left(\epsilon\right)\|f\|_{2}^{2}\|f\|_{\infty}^{2}.
\]
 By Lemmas \ref{lem:Other approximate efron stein decompositions for the Laplacians.}
and \ref{lem:L_4 Closeness of different Efron--Stein decompositions}
we have
\[
\mathbb{E}_{\left(x,y\right)\sim\mu_{S\cup T}}I_{T,y}^{2}\left[g_{S,x}\right]=\mathbb{E}_{z\sim\mu_{T\cup S}}I_{T\cup S,z}^{2}\left[g\right]+O_{k}\left(\epsilon^{2}\right)\|f\|_{2}^{2}\|f\|_{\infty}^{2}.
\]
 Hence, 
\[
\mathbb{E}_{x}\|g_{S,x}\|_{4}^{4}\le20^{d-\left|S\right|}\sum_{S'\supseteq S\left|S'\right|\le d}\left(4d\right)^{\left|S'\setminus S\right|}\mathbb{E}_{z\sim\mu_{S'}}\left(I_{S',z}^{\le d}\left[f\right]\right)^{2}+O_{k}\left(\epsilon\right)\|f\|_{2}^{2}\|f\|_{\infty}^{2}.
\]
 This gives
\[
\mathbb{E}_{x}\|D_{S,x}^{\le d}\left[f\right]\|_{4}^{4}\le2\cdot20^{d-\left|S\right|}\sum_{S'\supseteq S\left|S'\right|\le d}\left(4d\right)^{\left|S'\setminus S\right|}\mathbb{E}_{z\sim\mu_{S'}}\left(I_{S',z}^{\le d}\left[f\right]\right)^{2}+O_{k}\left(\epsilon\right)\|f\|_{2}^{2}\|f\|_{\infty}^{2}
\]
The proof is now completed by plugging this inequality in (\ref{eq:12-1}).
Indeed, we have 
\begin{align*}
\|f^{\le d}\|_{4}^{4} & \le2\cdot9^{d}\|f\|_{2}^{4}+O_{k}\left(\epsilon\|f\|_{2}^{2}\|f\|_{\infty}^{2}\right)\\
 & +\sum_{0<\left|S\right|\le d}\left(4d\right)^{\left|S\right|}\cdot2\cdot20^{d-\left|S\right|}\sum_{S'\supseteq S\left|S'\right|\le d}\left(4d\right)^{\left|S'\setminus S\right|}\mathbb{E}_{z\sim\mu_{S'}}\left(I_{S',z}^{\le d}\left[f\right]\right)^{2}\\
 & \le20^{d}\sum_{\left|S'\right|\le d}\left(4d\right)^{\left|S\right|}\mathbb{E}_{z\sim\mu_{S'}}\left(I_{S',z}^{\le d}\left[f\right]\right)^{2}+O_{k}\left(\epsilon\|f\|_{2}^{2}\|f\|_{\infty}^{2}\right).
\end{align*}
\end{proof}

\subsection{The case where $\|f\|_{\infty}$ is large}

Here we show a hypercontractive inequality whose error term does not
include the factor $\|f\|_{\infty}.$ This may be useful when $\|f\|_{\infty}$
is significantly larger than $\delta.$
\begin{thm}
\label{thm:our hypercontractivity for low degree functions} Suppose
that $f$ is $\left(d,\delta\right)$-global, then 
\[
\|f^{\le d}\|_{4}^{4}\le20^{d+1}\sum_{\left|S\right|\le d}\left(4d\right)^{\left|S\right|}\mathbb{E}_{x\sim\mu_{S}}I_{S,x}^{\le d}\left[f\right]^{2}+O_{k}\left(\epsilon^{2}\delta^{2}\right)\|f\|_{2}^{2}.
\]
\end{thm}

\begin{proof}
By applying Theorem \ref{thm:Globalness implies small 4-norm} with
$f^{\le d}$ rather then $f$ and using $\|f^{\le d}\|_{\infty}\le\delta$
we obtain 
\[
\|\left(f^{\le d}\right)^{\le d}\|_{4}^{4}\le20^{d}\sum_{\left|S\right|\le d}\left(4d\right)^{\left|S\right|}\mathbb{E}_{x\sim\mu_{S}}I_{S,x}^{\le d}\left[f^{\le d}\right]^{2}+O_{k}\left(\epsilon\delta^{2}\right)\|f\|_{2}^{2}.
\]
 The theorem now follows from Lemmas \ref{lem:f le d ES}, \ref{lem:L_4 Closeness of different Efron--Stein decompositions}
and \ref{lem:Many ways to compute influences }.
\end{proof}
\begin{thm}
\label{thm:Globalness implies small 4-norm} Let $\epsilon\le\epsilon_{0}\left(k\right)$
be sufficiently small. Suppose that $f$ is $\left(d,\delta\right)$-global.
Then we have 
\[
\|f^{\le d}\|_{4}^{4}\le\left(100d\right)^{d}\delta^{2}\|f^{\le d}\|_{2}^{2}+O_{k}\left(\delta^{2}\epsilon^{2}\|f\|_{2}^{2}\right).
\]
\end{thm}

\begin{proof}
By Theorem \ref{thm:our hypercontractivity for low degree functions},
Lemma \ref{lem:truncated influences are small}, and \ref{lem:Strong Parseval}
we have 
\begin{align*}
\|f^{\le d}\|_{4}^{4} & \le20^{d}\sum_{\left|S\right|\le d}\left(4d\right)^{d}\mathbb{E}_{x\sim\mu_{S}}\left(I_{S,x}^{\le d}\left[f\right]\right)^{2}+O_{k}\left(\epsilon^{2}\|f\|_{2}^{2}\|f\|_{\infty}^{2}\right)\\
 & \le20^{d}\sum_{\left|S\right|\le d}\left(8d\right)^{d+2}\delta^{2}\mathbb{E}_{x}\left[I_{S,x}^{\le d}\left[f\right]\right]+O_{k}\left(\epsilon\|f\|_{2}^{2}\|f\|_{\infty}^{2}\right)\\
 & \le20^{d}\sum_{\left|S\right|\le d}\left(8d\right)^{d+2}\delta^{2}\sum_{T\supseteq S,\left|T\right|\le d}\|f^{=T}\|_{2}^{2}+O_{k}\left(\epsilon\|f\|_{2}^{2}\|f\|_{\infty}^{2}\right)\\
 & \le\left(40d\right)^{d}\delta^{2}\sum_{\left|T\right|\le d}\|f^{=T}\|_{2}^{2}+O_{k}\left(\epsilon\|f\|_{2}^{2}\|f\|_{\infty}^{2}\right)\\
 & \le2\left(40d\right)^{d}\delta^{2}\|f^{\le d}\|_{2}^{2}+O_{k}\left(\epsilon\|f\|_{2}^{2}\|f\|_{\infty}^{2}\right).
\end{align*}
\end{proof}

\section{Applications\label{sec:Applications}}

In this section, we show our applications of the hypercontractive
inequality on high dimensional expanders, which we have shown in the
previous section. The applications follow in a fairly straightforward
way, and hence we present them with brevity.

\subsection{Global Boolean functions are concentrated on the high degrees.}

Fourier concentration results are widely useful in complexity theory
and learning theory. Our first application is a Fourier concentration
theorem for HDX. Namely, the following theorem shows that global Boolean
functions on $\epsilon$-HDX are concentrated on the high degrees,
in the sense that the $2$-norm of the restriction of a function to
its low-degree coefficients only constitutes a tiny fraction of its
total $2$-norm.
\begin{cor}
\label{cor:low degree part is small}If $f\colon V_{\left[k\right]}\to\left\{ 0,1\right\} $
is $\left(d,\delta\right)$-global and $\epsilon$ is sufficiently
small. Then 
\[
\|f^{\le d}\|_{2}^{2}\le\left(O_{k}\left(\sqrt{\epsilon}\right)+\left(200d\right)^{d}\delta^{\frac{1}{2}}\right)\|f\|_{2}^{2}.
\]
\end{cor}

\begin{proof}
By Lemma \ref{lem:Strong Parseval} we have 
\[
\|f^{\le d}\|_{2}^{2}=\left\langle f^{\le d},f\right\rangle -O_{k}\left(\epsilon\right)\|f\|_{2}^{2}.
\]
We also have by Theorem \ref{thm:Globalness implies small 4-norm}
\begin{align*}
\left\langle f^{\le d},f\right\rangle  & \le\|f^{\le d}\|_{4}\|f\|_{\frac{4}{3}}\\
 & \le\left(100d\right)^{d}\delta^{\frac{1}{2}}\sqrt{\|f^{\le d}\|_{2}}\|f\|_{\frac{4}{3}}+O_{k}\left(\sqrt{\epsilon}\|f\|_{2}^{\frac{1}{2}}\|f\|_{\infty}^{\frac{1}{2}}\|f\|_{\frac{4}{3}}\right).\\
 & \le2\left(100d\right)^{d}\delta^{\frac{1}{2}}\sqrt{\|f\|_{2}}\|f\|_{\frac{4}{3}}+O_{k}\left(\sqrt{\epsilon}\|f\|_{2}^{2}\right)\\
 & \le\left(200d\right)^{d}\delta^{\frac{1}{2}}\|f\|_{2}^{2}+O_{k}\left(\sqrt{\epsilon}\|f\|_{2}^{2}\right)
\end{align*}
\end{proof}
The Corollary completes the proof of Theorem \ref{thm:concentration on the high degrees}.

\subsection{Small-set expansion theorem}

Small set expansion is a fundamental property that is prevalent in
combinatorics and complexity theory. In the setting of the $\rho$-noisy
Boolean hypercube, the small set expansion theorem gives an upper
bound on $\mathrm{{Stab}_{\rho}}(1_{A})=\langle1_{A},T_{\rho}1_{A}\rangle=E[1_{A}(x)1_{A}(y)]$
for indicators $1_{A}$ of small sets $A$, which captures the probability
that a random walk starting at a point $x\in A$ remains in $A$,
hence showing that small sets are expanding. Our second application
is a small set expansion theorem for global functions on $\epsilon$-HDX,
captured via bounding the natural noise operator in this setting.
\begin{defn}
Let $\rho\in\left(0,1\right).$ Given $x\in V_{\left[k\right]}$ we
let $N_{\rho}\left(x\right)$ be the distribution where $y\sim N_{\rho}\left(x\right)$
is chosen by choosing a random set $S$ where each $i$ is in $S$
independently with probability $\rho,$ then choosing $z\sim\mu_{x_{S}}$
and setting $y=\left(x_{S},z\right).$ We then set 
\[
\mathrm{T}_{\rho}f\left(x\right)=\mathbb{E}_{y\sim N_{\rho}\left(x\right)}f.
\]
 Alternatively we can use the averaging operators to give the following
equivalent definition:

\[
\mathrm{T}_{\rho}:=\sum_{S\subseteq\left[k\right]}\rho^{\left|S\right|}\left(1-\rho\right)^{k-\left|S\right|}A_{S}\left[f\right].
\]
\end{defn}

We have the following formula for the noise operator, which is similar
to the one in the product space setting.
\begin{claim}
\label{lem:Formula for the noise operator} We have $\mathrm{T}_{\rho}f=\sum_{S}\rho^{\left|S\right|}f^{=S}.$
\end{claim}

\begin{proof}
We have 
\begin{align*}
\mathrm{T}_{\rho}f & =\sum_{S\subseteq\left[k\right]}\rho^{\left|S\right|}\left(1-\rho\right)^{k-\left|S\right|}A_{S}\left[f\right]\\
 & =\sum_{S\subseteq\left[k\right]}\rho^{\left|S\right|}\left(1-\rho\right)^{k-\left|S\right|}\sum_{T\subseteq S}f^{=T}\\
 & =\sum_{T\subseteq\left[k\right]}\sum_{S\supseteq T}\rho^{\left|S\right|}\left(1-\rho\right)^{k-\left|S\right|}\\
 & =\sum_{T\subseteq\left[k\right]}\rho^{\left|T\right|}f^{=T}.
\end{align*}
\end{proof}
Via a standard argument we have the following bound on the noise operator.
\begin{lem}
\label{lem:traditional upper bound} We have 
\[
\|\mathrm{T}_{\rho}f\|_{2}^{2}\le\|f^{\le d}\|_{2}^{2}+\left(\rho^{d}+O_{k}\left(\epsilon\right)\right)\|f\|_{2}^{2}.
\]
\end{lem}

\begin{proof}
This is immediate from Lemmas \ref{lem:Formula for the noise operator}
and \ref{lem:Strong Parseval}.
\end{proof}
Our small set expansion applications are as follows.
\begin{cor}[Small set expansion theorem]
 If $f\colon V_{\left[k\right]}\to\left\{ 0,1\right\} $ is $\left(d,\delta\right)$-global.
Then 
\[
\|\mathrm{T}_{\rho}f\|_{2}^{2}\le\left(\rho^{d}+\left(100d\right)^{d}\delta^{2}+O_{k}\left(\sqrt{\epsilon}\right)\right)\|f\|_{2}^{2}.
\]
\end{cor}

\begin{proof}
This follows immediately from Lemma \ref{lem:traditional upper bound}
and Corollary \ref{cor:low degree part is small}.
\end{proof}

\subsection{Kruskal--Katona theorem}

Our last application is an analogue of the Kruskal--Katona theorem
in the setting of high dimensional expanders. The Kruskal-Katona theorem
is a fundamental and widely-applied result in algebraic combinatorics,
which gives a lower bound on the size of the lower shadow of a set
$A$, denoted $\partial(A)=\{x\colon y\prec x,\;\text{for some}\;y\in A\}$. 

We first consider the natural up-down walk in our setting.
\begin{defn}
The operator corresponding to up-down random walk is 
\[
\mathrm{T}=\frac{1}{k}\sum_{i=1}^{k}A_{\left[k\right]\setminus\left\{ i\right\} }\left[f\right]=\sum_{S}\frac{k-\left|S\right|}{k}f^{=S}.
\]

By applying the approximate Parseval inequality (Lemma \ref{lem:Strong Parseval}),
we obtain the following claim.
\end{defn}

\begin{claim}
We have 
\[
\left\langle f-\mathrm{T}f,f\right\rangle \ge\frac{d}{k}\|f^{\ge d}\|_{2}^{2}-O_{k}\left(\epsilon\right)\|f\|_{2}^{2}.
\]
\end{claim}

By our Fourier concentration theorem, (Corollary \ref{cor:low degree part is small}),
we have the following lower bound on the $2$-norm of the high degree
part of $f$.
\begin{claim}
Let $\delta\le\left(200d\right)^{-2d},$ and $\epsilon\le\epsilon_{0}\left(k\right)$
be sufficiently small. If $f\colon V_{\left[k\right]}\to\left\{ 0,1\right\} $
is $\left(d,\delta\right)$-global. Then 
\[
\|f^{\ge d}\|_{2}^{2}\ge\frac{1}{2}\|f\|_{2}^{2}.
\]
\end{claim}

Combining the above claims we get the following.
\begin{claim}
Let $\delta\le\left(200d\right)^{-2d}.$ We have 
\[
\left\langle f-\mathrm{T}f,f\right\rangle \ge\frac{d}{2k}\|f\|_{2}^{2}
\]
\end{claim}

We are now ready to prove the Kruskal--Katona theorem in the setting
of high dimensional expanders.
\begin{cor}
Let $X$ be an $\epsilon$-HDX, for a sufficiently small $\epsilon>0$.
Let $\delta\le\left(200d\right)^{-d},$ and let $A\subseteq X\left(k-1\right)$
be $\left(d,\delta\right)$-global. Then 
\[
\mu\left(\partial(A)\right)\ge\mu\left(A\right)\left(1+\frac{d}{2k}\right).
\]
\end{cor}

\begin{proof}
Let $f=1_{A}.$ We have 
\begin{align*}
\left\langle f-\mathrm{T}f,f\right\rangle  & =\Pr_{\sigma\sim X\left(k-1\right)}\Pr_{\tau_{1},\tau_{2}\supset\sigma}\left[\tau_{1}\in A,\tau_{2}\notin A\right].\\
 & \le\Pr\left[\sigma\in\partial(A),f\left(\tau_{2}\right)\notin A\right]\\
 & =\mu\left(\partial(A)\right)-\mu\left(A\right).
\end{align*}
\end{proof}
\bibliographystyle{plain}
\bibliography{refs}

\begin{thebibliography}{10}

\bibitem{alev2019approximating}
Vedat~Levi Alev, Fernando~Granha Jeronimo, and Madhur Tulsiani.
\newblock Approximating constraint satisfaction problems on high-dimensional
  expanders.
\newblock In {\em 2019 IEEE 60th Annual Symposium on Foundations of Computer
  Science (FOCS)}, pages 180--201. IEEE, 2019.

\bibitem{anari2021entropic}
Nima Anari, Vishesh Jain, Frederic Koehler, Huy~Tuan Pham, and Thuy-Duong
  Vuong.
\newblock Entropic independence in high-dimensional expanders: Modified
  log-sobolev inequalities for fractionally log-concave polynomials and the
  ising model.
\newblock {\em arXiv preprint arXiv:2106.04105}, 2021.

\bibitem{anari2021spectral}
Nima Anari, Kuikui Liu, and Shayan~Oveis Gharan.
\newblock Spectral independence in high-dimensional expanders and applications
  to the hardcore model.
\newblock {\em SIAM Journal on Computing}, (0):FOCS20--1, 2021.

\bibitem{anari2019log}
Nima Anari, Kuikui Liu, Shayan~Oveis Gharan, and Cynthia Vinzant.
\newblock Log-concave polynomials ii: high-dimensional walks and an fpras for
  counting bases of a matroid.
\newblock In {\em Proceedings of the 51st Annual ACM SIGACT Symposium on Theory
  of Computing}, pages 1--12, 2019.

\bibitem{BHKS}
Mitali Bafna, Max Hopkins, Tali Kaufman, and Schahar Lovett.
\newblock Hypercontractivity on high dimensional expanders.
\newblock {\em (personal communication)}, 2021.

\bibitem{braverman2021invariance}
Mark Braverman, Subhash Khot, Noam Lifshitz, and Dor Minzer.
\newblock An invariance principle for the multi-slice, with applications.
\newblock {\em arXiv preprint arXiv:2110.10725}, 2021.

\bibitem{chen2021optimal}
Zongchen Chen, Kuikui Liu, and Eric Vigoda.
\newblock Optimal mixing of glauber dynamics: Entropy factorization via
  high-dimensional expansion.
\newblock In {\em Proceedings of the 53rd Annual ACM SIGACT Symposium on Theory
  of Computing}, pages 1537--1550, 2021.

\bibitem{DFLLV2021}
Neta Dafni, Yuval Filmus, Noam Lifshitz, Nathan Lindzey, and Marc Vinyals.
\newblock Complexity measures on symmetric group and beyond.
\newblock In {\em 12th Innovations in Theoretical Computer Science Conference
  (ITCS 2021)}, 2021.

\bibitem{DiksteinDFH18}
Yotam Dikstein, Irit Dinur, Yuval Filmus, and Prahladh Harsha.
\newblock Boolean {F}unction {A}nalysis on {H}igh-{D}imensional {E}xpanders.
\newblock In {\em LIPIcs-Leibniz International Proceedings in Informatics},
  volume 116. Schloss Dagstuhl-Leibniz-Zentrum fuer Informatik, 2018.

\bibitem{dikstein2020locally}
Yotam Dikstein, Irit Dinur, Prahladh Harsha, and Noga Ron-Zewi.
\newblock Locally testable codes via high-dimensional expanders.
\newblock {\em arXiv preprint arXiv:2005.01045}, 2020.

\bibitem{c3LTC}
Irit Dinur, Shai Evra, Ron Livne, Alex Lubotzky, and Shahar Mozes.
\newblock Locally testable codes with constant rate, distance, and locality.
\newblock In {\em (to appear)}, 2021.

\bibitem{dinur2020explicit}
Irit Dinur, Yuval Filmus, Prahladh Harsha, and Madhur Tulsiani.
\newblock Explicit sos lower bounds from high-dimensional expanders.
\newblock {\em arXiv preprint arXiv:2009.05218}, 2020.

\bibitem{DinurHKNT19}
Irit Dinur, Prahladh Harsha, Tali Kaufman, Inbal~Livni Navon, and Amnon
  Ta{-}Shma.
\newblock List {D}ecoding with {D}ouble {S}amplers.
\newblock In {\em Proceedings of the Thirtieth Annual {ACM-SIAM} Symposium on
  Discrete Algorithms, {SODA} 2019, San Diego, California, USA, January 6-9,
  2019}, pages 2134--2153, 2019.

\bibitem{DinurK17}
Irit Dinur and Tali Kaufman.
\newblock High {D}imensional {E}xpanders {I}mply {A}greement {E}xpanders.
\newblock In {\em Foundations of Computer Science (FOCS), 2017 IEEE 58th Annual
  Symposium on}, pages 974--985. IEEE, 2017.

\bibitem{DinurKKMS18a}
Irit Dinur, Subhash Khot, Guy Kindler, Dor Minzer, and Muli Safra.
\newblock On {N}on-{O}ptimally {E}xpanding {S}ets in {G}rassmann {G}raphs.
\newblock In {\em Proceedings of the 50th Annual {ACM} {SIGACT} Symposium on
  Theory of Computing, {STOC} 2018, Los Angeles, CA, USA, June 25-29, 2018},
  pages 940--951, 2018.

\bibitem{DinurKKMS18}
Irit Dinur, Subhash Khot, Guy Kindler, Dor Minzer, and Muli Safra.
\newblock Towards a {P}roof of the 2-to-1 {G}ames {C}onjecture?
\newblock In {\em Proceedings of the 50th Annual {ACM} {SIGACT} Symposium on
  Theory of Computing, {STOC} 2018, Los Angeles, CA, USA, June 25-29, 2018},
  pages 376--389, 2018.

\bibitem{ellis2019biased}
David Ellis, Nathan Keller, and Noam Lifshitz.
\newblock On a biased edge isoperimetric inequality for the discrete cube.
\newblock {\em Journal of Combinatorial Theory, Series A}, 163:118--162, 2019.

\bibitem{EKL21}
David Ellis, Guy Kindler, and Noam Lifshitz.
\newblock Hypercontractivity for global functions in the bilinear scheme.
\newblock {\em In preparation}, 2021.

\bibitem{evra2020decodable}
Shai Evra, Tali Kaufman, and Gilles Z{\'e}mor.
\newblock Decodable quantum ldpc codes beyond the square root distance barrier
  using high dimensional expanders.
\newblock In {\em 2020 IEEE 61st Annual Symposium on Foundations of Computer
  Science (FOCS)}, pages 218--227. IEEE, 2020.

\bibitem{Filmus16}
Yuval Filmus.
\newblock Friedgut-{K}alai-{N}aor {T}heorem for {S}lices of the {B}oolean
  {C}ube.
\newblock {\em Chicago J. Theor. Comput. Sci.}, 2016, 2016.

\bibitem{Filmus2021+}
Yuval Filmus.
\newblock Almost linear boolean functions on { S n } are almost unions of
  cosets.
\newblock Manuscript, 2021.

\bibitem{Filmus2021exposition}
Yuval Filmus.
\newblock {FKN} theorems for the biased cube and the slice.
\newblock Manuscript, 2021.

\bibitem{filmus2020hypercontractivity}
Yuval Filmus, Guy Kindler, Noam Lifshitz, and Dor Minzer.
\newblock Hypercontractivity on the symmetric group.
\newblock {\em arXiv preprint arXiv:2009.05503}, 2020.

\bibitem{FilmusKMW18}
Yuval Filmus, Guy Kindler, Elchanan Mossel, and Karl Wimmer.
\newblock Invariance {P}rinciple on the {S}lice.
\newblock {\em {TOCT}}, 10(3):11:1--11:37, 2018.

\bibitem{FilmusM16}
Yuval Filmus and Elchanan Mossel.
\newblock Harmonicity and {I}nvariance on {S}lices of the {B}oolean {C}ube.
\newblock In {\em 31st Conference on Computational Complexity, {CCC} 2016, May
  29 to June 1, 2016, Tokyo, Japan}, pages 16:1--16:13, 2016.

\bibitem{FOW2019}
Yuval Filmus, Ryan O'Donnell, and Xinyu Wu.
\newblock A log-{S}obolev inequality for the multislice, with applications.
\newblock In {\em Proceedings of the 10th Innovations in Theoretical Computer
  Science conference (ITCS'19)}, 2019.

\bibitem{golowich2021improved}
Louis Golowich.
\newblock Improved product-based high-dimensional expanders.
\newblock {\em arXiv preprint arXiv:2105.09358}, 2021.

\bibitem{gotlib2019testing}
Roy Gotlib and Tali Kaufman.
\newblock Testing odd direct sums using high dimensional expanders.
\newblock In {\em Approximation, Randomization, and Combinatorial Optimization.
  Algorithms and Techniques (APPROX/RANDOM 2019)}. Schloss
  Dagstuhl-Leibniz-Zentrum fuer Informatik, 2019.

\bibitem{hopkins2020high}
Max Hopkins, Tali Kaufman, and Shachar Lovett.
\newblock High dimensional expanders: Random walks, pseudorandomness, and
  unique games.
\newblock {\em arXiv preprint arXiv:2011.04658}, 2020.

\bibitem{jeronimo2021near}
Fernando~Granha Jeronimo, Shashank Srivastava, and Madhur Tulsiani.
\newblock Near-linear time decoding of ta-shma's codes via splittable
  regularity.
\newblock In {\em Proceedings of the 53rd Annual ACM SIGACT Symposium on Theory
  of Computing}, pages 1527--1536, 2021.

\bibitem{KahnKL88}
Jeff Kahn, Gil Kalai, and Nathan Linial.
\newblock {\em The {I}nfluence of {V}ariables on {B}oolean {F}unctions}.
\newblock IEEE, 1988.

\bibitem{KaufmanL14}
Tali Kaufman and Alexander Lubotzky.
\newblock High {D}imensional {E}xpanders and {P}roperty {T}esting.
\newblock In {\em Innovations in Theoretical Computer Science, ITCS'14,
  Princeton, NJ, USA, January 12-14, 2014}, pages 501--506, 2014.

\bibitem{KaufmanM17}
Tali Kaufman and David Mass.
\newblock High {D}imensional {R}andom {W}alks and {C}olorful {E}xpansion.
\newblock In {\em LIPIcs-Leibniz International Proceedings in Informatics},
  volume~67. Schloss Dagstuhl-Leibniz-Zentrum fuer Informatik, 2017.

\bibitem{KaufmanO18a}
Tali Kaufman and Izhar Oppenheim.
\newblock Construction of {N}ew {L}ocal {S}pectral {H}igh {D}imensional
  {E}xpanders.
\newblock In {\em Proceedings of the 50th Annual {ACM} {SIGACT} Symposium on
  Theory of Computing, {STOC} 2018, Los Angeles, CA, USA, June 25-29, 2018},
  pages 773--786, 2018.

\bibitem{KaufmanO20}
Tali Kaufman and Izhar Oppenheim.
\newblock High {O}rder {R}andom {W}alks: {B}eyond {S}pectral {G}ap.
\newblock {\em Combinatorica}, pages 1--37, 2020.

\bibitem{keevash2019hypercontractivity}
Peter Keevash, Noam Lifshitz, Eoin Long, and Dor Minzer.
\newblock Hypercontractivity for global functions and sharp thresholds.
\newblock {\em arXiv preprint arXiv:1906.05568}, 2019.

\bibitem{KLLM}
Peter Keevash, Noam Lifshitz, Eoin Long, and Dor Minzer.
\newblock Forbidden intersections for codes.
\newblock {\em arXiv preprint arXiv:2103.05050}, 2021.

\bibitem{KhotMS17}
Subhash Khot, Dor Minzer, and Muli Safra.
\newblock On {I}ndependent {S}ets, 2-to-2 {G}ames, and {G}rassmann {G}raphs.
\newblock In {\em Proceedings of the 49th Annual {ACM} {SIGACT} Symposium on
  Theory of Computing, {STOC} 2017, Montreal, QC, Canada, June 19-23, 2017},
  pages 576--589, 2017.

\bibitem{KhotMS18}
Subhash Khot, Dor Minzer, and Muli Safra.
\newblock Pseudorandom {S}ets in {G}rassmann {G}raph {H}ave {N}ear-{P}erfect
  {E}xpansion.
\newblock In {\em 59th {IEEE} Annual Symposium on Foundations of Computer
  Science, {FOCS} 2018, Paris, France, October 7-9, 2018}, pages 592--601,
  2018.

\bibitem{lifshitz2019noise}
Noam Lifshitz and Dor Minzer.
\newblock Noise sensitivity on the p-biased hypercube.
\newblock In {\em 2019 IEEE 60th Annual Symposium on Foundations of Computer
  Science (FOCS)}, pages 1205--1226. IEEE, 2019.

\bibitem{liu2019high}
Siqi Liu, Sidhanth Mohanty, and Elizabeth Yang.
\newblock High-dimensional expanders from expanders.
\newblock {\em arXiv preprint arXiv:1907.10771}, 2019.

\bibitem{lubotzky2018high}
Alexander Lubotzky.
\newblock High dimensional expanders.
\newblock In {\em Proceedings of the International Congress of Mathematicians:
  Rio de Janeiro 2018}, pages 705--730. World Scientific, 2018.

\bibitem{ODonnell}
Ryan O'Donnell.
\newblock {\em Analysis of boolean functions}.
\newblock Cambridge University Press, 2014.

\bibitem{ODonnellK13}
Ryan {O'Donnell} and Karl Wimmer.
\newblock {KKL}, {K}ruskal--{K}atona, and {M}onotone {N}ets.
\newblock {\em SIAM Journal on Computing}, 42(6):2375--2399, 2013.

\bibitem{parzanchevski2017mixing}
Ori Parzanchevski.
\newblock Mixing in high-dimensional expanders.
\newblock {\em Combinatorics, Probability and Computing}, 26(5):746--761, 2017.

\end{thebibliography}

\end{document}